\documentclass[11pt,letterpaper]{article}

\usepackage[utf8]{inputenc}
\usepackage{lmodern} %
\usepackage{authblk}
\usepackage[bookmarks,colorlinks,breaklinks]{hyperref}
\hypersetup{urlcolor=blue, colorlinks=true, citecolor=green!50!black, linkcolor=blue}
\usepackage{multirow}
\usepackage[T1]{fontenc} %
\usepackage{amsmath}
\usepackage{amsfonts}
\usepackage{amssymb}
\usepackage{amsthm}
\usepackage{thmtools}
\usepackage{xcolor} %
\usepackage[lined,boxed,ruled,norelsize,algo2e,linesnumbered,noend]{algorithm2e} %
\usepackage{cleveref}
\usepackage{graphicx} %
\usepackage{geometry} %
\usepackage{enumitem} %
\usepackage{url} %
\usepackage{tikz} %

\ifdefined\DEBUG
	\usepackage{showlabels} %

\else
\fi

\declaretheorem[numberwithin=section,refname={Theorem,Theorems},Refname={Theorem,Theorems}]{theorem}
\declaretheorem[numberlike=theorem]{lemma}

\declaretheorem[numberlike=theorem]{corollary}
\declaretheorem[numberlike=theorem]{definition}

\declaretheorem[numberlike=theorem]{fact}

\newcommand{\curly}{\mathrel{\leadsto}}
\renewcommand{\path}{\mathrel{\leadsto}}

\newcommand{\R}{\mathbb{R}}
\newcommand{\N}{\mathbb{N}}

\renewcommand{\tilde}{\widetilde}

\newcommand{\polylog}{\operatorname{polylog}}
\newcommand{\dist}{\operatorname{dist}}
\newcommand{\len}{\operatorname{len}}

\usepackage{pgffor} %
\foreach \x in {A,...,Z}{%
	\expandafter\xdef\csname m\x\endcsname{\noexpand\mathbf{\x}}
}
\foreach \x in {A,...,Z}{%
	\expandafter\xdef\csname cal\x\endcsname{\noexpand\mathcal{\x}}
}

\newcommand{\s}{a}

\makeatletter
\renewcommand{\paragraph}{%
	\@startsection{paragraph}{4}%
	{\z@}{1.25ex \@plus 1ex \@minus .2ex}{-1em}%
	{\normalfont\normalsize\bfseries}%
}
\makeatother

\usetikzlibrary{decorations.pathmorphing}
\usetikzlibrary{snakes}

\usetikzlibrary {graphs}
\usetikzlibrary{arrows}
\usetikzlibrary{arrows.meta,arrows}

\usepackage[
    sorting=nyt,
	style=alphabetic,
    maxalphanames=3,
    minalphanames=3,
    maxcitenames=3,
    minnames=1,
    maxbibnames=20,
	backref=true,
	doi=true,
	url=true,
	backend=biber
]{biblatex}

\DeclareCiteCommand{\citem}
    {}
    {\mkbibbrackets{\bibhyperref{\usebibmacro{postnote}}}}
    {\multicitedelim}
    {}

\newcommand{\exact}{1.823}
\newcommand{\exactQuery}{1.747}
\newcommand{\exactPre}{2.626}

\newcommand{\exactS}{0.25308461}
\newcommand{\exactNu}{0.2050319}
\newcommand{\exactMu}{0.45811651}

\newcommand{\directed}{1.816}
\newcommand{\directedQuery}{1.741}
\newcommand{\directedPre}{2.633}

\newcommand{\directedS}{0.25955649}
\newcommand{\directedNu}{0.22133053}
\newcommand{\directedMu}{0.48088702}
\newcommand{\undirected}{1.72} %
\newcommand{\undirectedPre}{2.564}

\title{Fully Dynamic Shortest Path Reporting\\ Against an Adaptive Adversary} %
\author[1]{Anastasiia Alokhina}
\affil[1]{Georgia Institute of Technology}
\author[1]{Jan van den Brand}

\begin{document}
\pagenumbering{roman}
\maketitle

\begin{abstract}
    
Algebraic data structures are the main subroutine for maintaining 
distances in fully dynamic graphs in subquadratic time.
However, these dynamic algebraic algorithms generally cannot maintain the shortest paths, especially against adaptive adversaries.
We present the first fully dynamic algorithm that maintains the shortest paths against an adaptive adversary in subquadratic update time.
This is obtained via a combinatorial reduction that  allows reconstructing the shortest paths with only a few distance estimates. Using this reduction, we obtain the following:

On weighted directed graphs with real edge weights in $[1,W]$, we can maintain $(1+\epsilon)$-approximate shortest paths in $\tilde{O}(n^{\directed}\epsilon^{-2}\log W)$ update and $\tilde{O}(n^{\directedQuery}\epsilon^{-2}\log W)$ query time. This improves upon the approximate distance data structures from 
\citem[v.d.Brand, Nanongkai; FOCS'19]{BrandN19},
which only returned a distance estimate, by matching their complexity and returning an approximate shortest path.

On unweighted directed graphs, we can maintain exact shortest paths in $\tilde{O}(n^{\exact})$ update and $\tilde{O}(n^{\exactQuery})$ query time. This improves upon 
\citem[Bergamaschi, Henzinger, P.Gutenberg, V.Williams, Wein; SODA'21]{BergamaschiHGWW21} 
who could report the path only against oblivious adversaries. We improve both their update and query time while also handling adaptive adversaries.

On unweighted undirected graphs, our reduction holds not just against adaptive adversaries but is also deterministic.
We maintain a $(1+\epsilon)$-approximate $st$-shortest path in $O(n^{1.529}/\epsilon^2)$ time per update,
and $(1+\epsilon)$-approximate single source shortest paths in $O(n^{1.764}/\epsilon^2)$ time per update.
Previous %
deterministic results by 
\citem[v.d.Brand,Nazari,Forster; FOCS'22]{BrandFN22} 
could only maintain distance estimates but no paths.

\end{abstract}
\newpage
\tableofcontents
\newpage
\pagenumbering{arabic}

\newcommand{\OrAprxUpdNDep}{n^{\omega(1, 1, \nu + \s)-\nu}}

\newcommand{\OrAprxQueryHHNDep}{n^{\omega(1-\s,1-\s,\nu + \s)}}
\newcommand{\OrAprxQueryVHNDep}{n^{\omega(1,1-\s,\nu + \s)}}

\newcommand{\OrExactQueryNDep}{n^{\s + \mu}}
\newcommand{\OrExactQueryHHNDep}{n^{\omega(1-\s,1-\s,\nu) + \s}}
\newcommand{\OrExactQueryVHNDep}{n^{\omega(1,1-\s,\nu) + \s}}
\newcommand{\OrExactUpdNDep}{(n^{1+\mu+\s} + n^{\omega(1, 1, \mu)-\mu + \s})}

\newcommand{\thmAprxDirUpd}{
\left(
      n^{\omega(1,1, \nu + \s)-\nu} 
      + 
      n^{1+\mu+\s}
      +
      n^{\omega(1,1, \mu)-\mu+\s} 
      +
      \OrAprxQueryVHNDep
      \right)\epsilon^{-2} \log W
}

\newcommand{\thmAprxDirQuery}{
       \left( 
           n^{2 - \s} 
         + n^{1 + \nu+\s } 
         + n^{1 + \s + \mu} \right)  \epsilon^{-2} \log W 
}

\newcommand{\thmAprxUndirUpd}{
    \left(
      n^{\omega(1,1, \nu + \s)-\nu} 
      + 
      n^{1+\mu+\s}
      +
      n^{\omega(1,1, \mu)-\mu+\s} 
      +
      n^{\omega(1 - \s,1 - \s, \nu + \s)}
    \right) \epsilon^{-2} \log W
    }

\newcommand{\thmAprxUndirQuery}{
       \left( 
     n^{1 + \nu}+n^{2 - 2\s} 
     + n^{1 + \s + \mu} \right)  \epsilon^{-2} \log W
}

\newcommand{\thmExactDirUpd}{
\left(
      n^{\omega(1,1, \nu + \s)-\nu} 
      + 
      n^{1+\mu+\s}
      +
      n^{\omega(1,1, \mu)-\mu+\s} 
      +
      n^{\omega(1 - \s,1 - \s, \mu)+\s} 
      +
     n^{\omega(1,1-\s, \nu+\s)} \right) {W}
}

\newcommand{\thmExactDirQuery}{
       \left( 
       n^{2 - \s} 
     + n^{1 + \nu + \s} 
     + n^{1 + \s + \mu} \right)  {W}
}

\newcommand{\thmExactUndirUpd}{
    \right(
      n^{1+\mu+\s}
      +
      n^{\omega(1,1, \mu)-\mu+\s} 
      +
      n^{\omega(1 - \s,1 - \s, \mu) + \s}
      \left){W}
    }

\newcommand{\thmExactUndirQuery}{
\left( 
       n^{2 - 2\s} 
     + n^{1 + \s + \mu} \right)  {W}
}

\section{Introduction}

In the dynamic shortest path problem, the task is to create a data structure that maintains the shortest paths in a given graph $G=(V, E)$ undergoing edge insertions and deletions.
This problem has a rich history (\cite{%
EvenS81,%
KleinS98,%
King99,%
DemetrescuI02,%
Thorup04,%
DemetrescuI04,%
Sankowski05,%
Thorup05,%
DemetrescuI06,%
BaswanaHS07,%
Bernstein09,%
RodittyZ11,%
BernsteinR11,%
RodittyZ12,%
AbrahamCG12,%
HenzingerKN13,%
HenzingerKN14,%
HenzingerKN14-2,%
AbrahamCDGW16,%
BernsteinC16,%
Bernstein16,%
BernsteinC17,%
AbrahamCK17,%
HenzingerKN18,%
BernsteinGW20,%
GutenbergW20,%
GutenbergW20a,%
BernsteinGS20,%
GutenbergW20b,%
GutenbergWW20,%
ChuzhoyS21,%
ChechikZ21,%
BernsteinGS21,%
Chuzhoy21,%
BergamaschiHGWW21,%
KarczmarzMS22,%
KyngMG22,
ChechikZ23})
with various variants being studied, such as partially dynamic (supporting only edge insertions or only edge deletions), fully dynamic (supporting both insertions and deletions), maintaining single-source shortest paths, all-pairs shortest path, or $st$-shortest path.
In this work, we focus on high-accuracy (i.e.~exact or $(1+\epsilon)$-approximate) shortest paths, as opposed to polylog or large constant factor approximations as studied in, e.g., \cite{AbrahamCT14,BaswanaKS12,BodwinK16,BernsteinBGNSS022,ChuzhoyZ23}.

A trivial solution for dynamic shortest path would be to run Dijkstra's algorithm from scratch in $O(n^2)$ time whenever distance information is queried.
It was shown by Abboud and V.Williams \cite{AbboudW14} that for %
shortest paths, beating $O(n^2)$ update and query time %
is only possible
when using algebraic techniques\footnote{In \cite{AbboudW14} it was shown that no dynamic algorithm could beat $O(n^{2-\epsilon})$ update time unless one can multiply two $n\times n$ matrices in $O(n^{3-\epsilon})$ time. 
Algorithms that make use of fast matrix multiplication are referred to as ``algebraic''.}.
Historically, while being the only tools for beating $O(n^2)$ update and query time, such dynamic algebraic algorithms have the downside that they only maintain the distance \cite{Sankowski05,BrandNS19,BrandN19} but not the shortest path itself.
Only very recently, the first progress has been made in maintaining the shortest path \cite{BergamaschiHGWW21} on unweighted graphs
in subquadratic time.
However, the path could only be maintained under the oblivious adversary assumption, i.e.~future updates are not allowed to depend on the shortest path previously returned by the data structure.

Removing the oblivious adversary assumption has received a lot of attention in the area of dynamic algorithms 
because, for many use cases the updates are adaptive (i.e.~depend on previous results).
One such example would be when a dynamic algorithm is used as a subroutine inside another algorithm
(e.g.~\cite{Madry10,ChuzhoyK19,BrandLSS20,ChuzhoyS21,BernsteinGS21,BernsteinBGNSS022,ChenKLPGS22}). 
Adaptivity issues also occur when the dynamic algorithm is used to analyze an interactive system, and the interactions with the system are chosen based on the output of the dynamic algorithm. 
Consider, for example, a map service that provides information about the fastest route given current traffic conditions. If the map service redirects its users away from traffic jams, then this will affect the current traffic conditions, so the input to the dynamic algorithm (map service) is adaptive.

Since assuming an oblivious adversary can be very restrictive or unreasonable depending on the application, removing this assumption has a rich history in dynamic algorithms \cite{BhattacharyaHN16,%
BernsteinC16,%
Bernstein17,%
NanongkaiSW17,%
ChuzhoyK19,%
Wajc20,%
GutenbergW20a,%
GutenbergWW20,%
EvaldFGW21,%
Chuzhoy21,%
BeimelKM+21,%
BrandFN22,%
BernsteinBGNSS022,%
KarczmarzMS22}
but also in other areas that analyze interactive systems, such as statistics \cite{HardtU14,RogersRST16,JungLN0SS20,BassilyNSSSU21,KontorovichSS22}, machine (un-)learning \cite{GuptaJNRSW21,Neel0S21}, or streaming algorithms \cite{WoodruffZ21,AlonBDMNY21,KaplanMNS21,AttiasCSS21,BenJWY22,BenEO22,ChakrabartiGS22,AjtaiBJSSWZ22,Cohen0NSSS22,CohenNSS22}.

We present the first high-accuracy dynamic algorithms that maintain shortest paths against an adaptive adversary in subquadratic time.
Despite solving a harder setting, our complexities match or improve that of previous work:
\begin{itemize}[nosep]
    \item 
    For $(1+\epsilon)$-approximate shortest paths on weighted graphs, we match the update and query time of the fastest dynamic algorithm that maintains the $(1+\epsilon)$-approximate distance against an adaptive adversary \cite{BrandN19}, but we also maintain the path itself.
    \item
    For exact unweighted shortest paths, we improve the update and query time of \cite{BergamaschiHGWW21} which could maintain the path only against oblivious adversaries and extend it to adaptive adversaries.
    \item
    On unweighted undirected graphs, \cite{BergamaschiHGWW21} additionally presented a faster algorithm for $(1+\epsilon)$-approximate shortest paths against oblivious adversaries.
    We match their update time and not only extend it to adaptive adversaries but provide a deterministic dynamic algorithm. 
    Previous deterministic results could only maintain the distance \cite{BrandFN22}, so our work is the first deterministic result for reporting paths in subquadratic time. 
\end{itemize}

\noindent
We remark that the only previous progress on maintaining paths in subquadratic time against adaptive adversaries had been made for the case of reachability, i.e.~maintaining \emph{any} path, not necessarily the shortest one. 
The dynamic algorithm by %
\cite{KarczmarzMS22} can maintain reachability with path reporting against an adaptive adversary in $O(n^{1.75+o(1)})$ amortized update time\footnote{In \cite{KarczmarzMS22} this was stated as $O(n^{1+5/6})$ but recent advances in decremental strongly connected components \cite{ChenKLPGS22} imply an $O(n^{1.75+o(1)})$ bound.\label{foot:reachability}}.
We can maintain the shortest path in directed graphs against an adaptive adversary, instead of just any path, in worst-case update time.

\subsection{Our Results and Comparison to Previous Work}
\label{sec:intro:results}

\begin{table}[]
    \centering
    \def\arraystretch{1.5}
    \begin{tabular}{|c|c|c|c|c|}
\hline 
Ref. & Update & Query & Path & Directed\tabularnewline
\hline 
\hline 
\multicolumn{5}{|c|}{$(1+\epsilon)$-approx., weighted}\tabularnewline
\hline 
\cite{BrandN19} & $n^{1.816}\frac{\log W}{\epsilon^2}$ & $n^{1.816}\frac{\log W}{\epsilon^2}$ & $\times$ & $\checkmark$\tabularnewline
\hline 
Thm~\ref{thm:intro:undirectedweighted} & $n^{\directed} \frac{\log W}{\epsilon^2}$ & $n^{\directedQuery}\frac{\log W}{\epsilon^2}$ & $\checkmark$ & $\checkmark$\tabularnewline
\hline 
Thm~\ref{thm:intro:directedweighted} & $n^{\undirected}\frac{\log W}{\epsilon^2}$ & $n^{\undirected}\frac{\log W}{\epsilon^2}$ & $\checkmark$ & $\times$\tabularnewline
\hline 
\multicolumn{5}{|c|}{$(1+\epsilon)$-approx., unweighted}\tabularnewline
\hline
\hspace{-5pt}\cite{BergamaschiHGWW21}\hspace{-5pt} & $n^{1.529+o_\epsilon(1)}$ & $n^{1+o_\epsilon(1)}$ & $\checkmark$ & $\times$ 
\tabularnewline
\hline 
\cite{BrandFN22} & $n^{1.406}\epsilon^{-2}$ & $n^{1.406}\epsilon^{-2}$ & $\times$ & $\times$\tabularnewline
\hline 
Thm~\ref{thm:intro:unweighted:st} & $n^{1.529}\epsilon^{-2}$ & $n^{1.529}\epsilon^{-2}$ & $\checkmark$ & $\times$\tabularnewline
\hline
\end{tabular}
\begin{tabular}{|c|c|c|c|}
\hline 
Ref. & Update & Query & Path \tabularnewline
\hline 
\hline 
\multicolumn{4}{|c|}{exact, directed}\tabularnewline
\hline 
\cite{Sankowski05} & $Wn^{1.897}$ & $Wn^{1.529}$ & $\times$ \tabularnewline
\hline 
\cite{BrandNS19} & $Wn^{1.724}$ & $Wn^{1.724}$ & $\times$ \tabularnewline
\hline 
\hspace{-5pt}\cite{BergamaschiHGWW21}\hspace{-5pt} & $Wn^{1.897}$ & $Wn^{1.897}$ & $\checkmark$ \tabularnewline
\hline
\cite{BrandFN22} & $Wn^{1.704}$ & $Wn^{1.704}$ & $\times$ \tabularnewline
\hline 
Thm~\ref{thm:intro:directedintweighted} & $Wn^{\exact}$ & $Wn^{\exactQuery}$ & $\checkmark$ \tabularnewline
\hline 
\multicolumn{4}{c}{}\tabularnewline
\multicolumn{4}{c}{}
\end{tabular}
    \caption{Comparison previous work on fully dynamic shortest paths with edge updates and pair queries.
    The $o_\epsilon(1)$ in \cite{BergamaschiHGWW21} is for constant $\epsilon>0$.
    Previous path reporting results require the oblivious adversary assumption.
    }
    \label{tbl:distance_results}
\end{table}

To maintain distances in subquadratic time, algebraic techniques are required \cite{AbboudW14}.
Historically, such algebraic data structures could only maintain the distances in a dynamic graph but not the corresponding shortest paths. Only recently have Bergamaschi, Henzinger, P.Gutenberg, V.Williams and Wein \cite{BergamaschiHGWW21} constructed the first extension to shortest path, but their dynamic algorithm works only against oblivious adversaries.
Since algebraic tools are so powerful in maintaining distances, a natural question would be how to efficiently extract the path from distance information.
In this work, we present such reductions. We show that given few distance queries, we can reconstruct the shortest path. 
These reductions work against adaptive adversaries and in the special case of unweighted undirected graphs, it is even deterministic.
We will first present the implications of our reductions and then discuss the technical idea of our reduction.

Our dynamic algorithms support a trade-off between update and query time. However, as our results heavily rely on fast (rectangular) matrix multiplication \cite{Williams12,Gall14,AlmanW21,DuanWZ22,GallU18}, the exact trade-off is complicated to state.
For simplicity, we state  all our results for a choice of parameters that minimizes the update time.
For each theorem, we provide a reference to the detailed version in \Cref{sec:fullalgo} that states the precise update vs query time trade-off.

\paragraph{Path Reporting against Adaptive Adversaries}
Before focusing on graphs with real edge weights, we present a result that works for small integer weights.

\begin{theorem}[{Directed, Small Weights, Exact, \Cref{thm:exact-dir}}]
\label{thm:intro:directedintweighted}
There exist fully dynamic algorithms maintaining exact shortest paths on directed graphs with integer edge weights in $[1, W]$, supporting edge updates and distance queries for any vertex pair.
The worst-case update time
is $\tilde O(n^{\exact}W)$, the query time is $\tilde O(n^{\exactQuery} W)$ and the preprocessing time is $\tilde O(n^{\exactPre}W)$.
The dynamic algorithm is randomized and correct w.h.p., and works against an adaptive adversary.
\end{theorem}

The only previous dynamic algorithm that can maintain the exact shortest path in subquadratic time is by Bergamaschi et al.~\cite{BergamaschiHGWW21} for directed unweighted graphs (or smaller integer weights). 
Their update and query complexity is $O(n^{1.897}W)$.
We improve upon both update and query time to $\tilde O(n^{\exact}W)$ and $\tilde O(n^{\exactQuery} W)$ respectively, despite working in the harder adaptive adversary model.
A trade-off between update and query time is possible and stated in \Cref{thm:exact-dir}.
The same $O(n^{1.897}W)$ complexity as in \cite{BergamaschiHGWW21} was first achieved by Sankowski \cite{Sankowski05}, whose dynamic algorithm returned only the distance but no path. 
The only faster dynamic algorithm for maintaining the exact distance (but not the path) has $O(n^{1.703}W)$ update and query time \cite{BrandNS19}.

We remark that the dynamic algorithm by Karczmarz, Mukherjee and Sankowski \cite{KarczmarzMS22} can maintain reachability with path reporting against an adaptive adversary. Note that here the task is to return \emph{any} path, not necessarily the shortest one.
On DAGs, they can maintain $st$-reachability in $O(n^{1.529})$ worst-case time and single source reachability in $O(n^{1.765})$ worst-case time, while also reporting a connecting path (or a tree in case of single source). 
To extend these results to general graphs, their update time becomes \emph{amortized} and in case of $st$-reachability also increases$^{\ref{foot:reachability}}$ to $O(n^{1.75+o(1)})$ update time.
Our dynamic shortest path algorithm \Cref{thm:intro:directedintweighted} works on directed graphs, so it can also maintain reachability with path reporting in %
\emph{worst-case} update time.

So far, we have only focused on unweighted graphs (or with small integer weights). 
Now, we present results that hold on real weighted graphs and return a $(1+\epsilon)$-approximate shortest path.
We remark that it is unlikely for an \emph{exact} dynamic algorithm with subquadratic update time to exist because that would contradict the APSP conjecture \cite{AbboudW14}.

\begin{theorem}[{Weighted, Approximate, \Cref{thm:aprx-dir,thm:aprx-undir}}]
\label{thm:intro:directedweighted}
\label{thm:intro:undirectedweighted}
There exist fully dynamic algorithms maintaining $(1+\epsilon)$-approximate shortest paths on weighted graphs with real edge weights in $[1, W]$, supporting edge updates and distance queries for any vertex pair.
\begin{itemize}[nosep]
    \item On \emph{directed} graphs, the worst-case update and query time
    are $\tilde O(n^{\directed}\epsilon^{-2}\log W)$ and \\$\tilde O(n^{\directedQuery}\epsilon^{-2}\log W)$ respectively. The preprocessing time is $\tilde O(n^{\directedPre}\epsilon^{-2} \log W)$.
    \item On \emph{undirected} graphs, the worst-case update and query time are $\tilde O(n^{\undirected}\epsilon^{-2}\log W)$. The preprocessing time is $\tilde O(n^{\undirectedPre}\epsilon^{-2} \log W)$.
\end{itemize}
The dynamic algorithms are randomized, correct w.h.p. and work against an adaptive adversary.
\end{theorem}

No other known fully dynamic algorithms maintain a ($(1+\epsilon)$-approximately) shortest path in weighted graphs in subquadratic time. 
While there are many dynamic algorithms that maintain paths on weighted graphs, they either maintain large constant or polylogarithmic approximation factors (e.g., \cite{AbrahamCT14,BaswanaKS12,BodwinK16,BernsteinBGNSS022,ChuzhoyZ23}), 
or study harder problems such as APSP and have $\Omega(n^2)$ update time (e.g., \cite{DemetrescuI04,Thorup05,AbrahamCK17,GutenbergW20a,ChechikZ23}).
If we ignore the task of returning a path and focus just on returning the (approximate) distance, then the fastest known dynamic algorithm 
on weighted graphs is by v.d.Brand and Nanongkai \cite{BrandN19}, which can maintain $(1+\epsilon)$-approximate $st$-distances in $\tilde O(n^{\directed}\epsilon^{-2}\log W)$ update and query time\footnote{In \cite{BrandN19}, this was stated as $\tilde O(n^{1.823}/\epsilon^2 \log W)$ update and query time, but their choice of trade-off parameters was not optimal.}.
Our result \Cref{thm:intro:directedweighted} matches this complexity from \cite{BrandN19} but can return the path itself. (We have the same dependence on fast matrix multiplication, so we match their result also for all future improvements on fast matrix multiplication.)

\paragraph{Deterministic Results on Unweighted Undirected Graphs}

In \cite{BrandFN22}, v.d.Brand, Forster and Nazari showed that maintaining approximate distance estimates can made deterministic when the input graph is unweighted and undirected.
Using their techniques, we can also provide a deterministic variants of our dynamic path reporting data structures on unweighted undirected graphs.

\begin{theorem}[{Undirected, Unweighted, Approximate, \Cref{lem:unweighted:st}}]
\label{thm:intro:unweighted:st}
There exists a fully dynamic algorithm that maintains a $(1+\epsilon)$-approximate $st$-shortest path
for unweighted undirected graphs.
The worst-case update and query time is $O(n^{1.529}\epsilon^{-2}\log\epsilon^{-1})$ and the preprocessing time is $O(n^{2.372}\epsilon^{-2}\log\epsilon^{-1})$.
The dynamic algorithm is deterministic.
\end{theorem}

The only previous dynamic algorithms for maintaining $(1+\epsilon)$-approximate shortest paths on unweighted undirected graphs
are by Bergamaschi et al.~\cite{BergamaschiHGWW21}.
Their algorithm could maintain shortest paths in $O_\epsilon(n^{1.529})$ update and $n^{1+o(1)}$ query time, but the dynamic algorithm is randomized and required the oblivious adversary assumption.
We not only match their update time and extend the dynamic algorithm to adaptive adversaries, but our result is also deterministic.

To support single-source queries, we can extend our dynamic algorithm as follows. Here a query returns an approximate shortest paths tree.  

\begin{theorem}[{Undirected, Unweighted, Approximate, \Cref{lem:unweighted:singlesource}}]
\label{thm:intro:unweighted:sssp}
There exists a fully dynamic algorithm that maintains $(1+\epsilon)$-approximate single-source shortest paths
for unweighted undirected graphs.
The worst-case update and query time is $O(n^{1.764}\epsilon^{-2}\log\epsilon^{-1})$ and the preprocessing time is $O(n^{2.609}\epsilon^{-2}\log\epsilon^{-1})$.
The dynamic algorithm is deterministic. 
\end{theorem}

\paragraph{Reduction from Paths to Distances} %

We show reductions from the dynamic shortest path to dynamic distances.
Our algorithm only assumes blackbox access to a dynamic algorithm that can maintain exact and approximate distances and uses this distance information to reconstruct the shortest path. 
The complexities stated in \Cref{sec:intro:results} are based on using the dynamic distance algorithms by \cite{Sankowski05,BrandN19} (\Cref{san-bn}). 
If in the future more efficient dynamic distance algorithms are developed, then our results will become faster as well.

Dynamic algebraic algorithms are very efficient for maintaining short distances but become slower as the maximum distance increases. This kind of data structure maintains ``$h$-bounded distances'', i.e.~for any $s,t,\in V$, they return the correct $st$-distance if the distance is at most $h$, and otherwise, it returns $\infty$.

\begin{restatable}[{Simplified version of \Cref{thm:exact-dir-gen}}]{theorem}{simplePathConstruction}\label{thm:intro:reduction} Let $h$ be a parameter between $1$ and $n$. Assume we are given a distance oracle on a directed graph $G=(V,E)$ with integer edge weights bounded by $W$. Given any two sets $S,T\subset V$, the oracle returns $Wh$-bounded distances for each pair in $S\times T$. 
Let $Q_\alpha(|S|,|T|)$ be the complexity of the oracle for returning an $\alpha$-approximation of the $Wh$-bounded distances for $S,T\subset V$.

Then given $s,t\in V$ we can reconstruct an exact $st$-shortest path in time
$$
\tilde{O}(Q_1(n/h,n/h) + Q_2(n,n/h) + n\cdot Q_1(1,1) \log (nW) ).
$$
\end{restatable}
\Cref{thm:intro:reduction} is a simplified version of the reduction that yields \Cref{thm:intro:directedintweighted} when combining with the dynamic distance algorithms from \cite{Sankowski05,BrandN19}.
We also prove similar reductions for $(1+\epsilon)$-approximate shortest paths on weighted graphs, which increases the complexity by an $\tilde{O}(\epsilon^{-1})$ factor (\Cref{thm:aprx-dir-gen}).

Notice that we only need a crude $2$-approximation\footnote{In fact, any constant factor approximation would suffice.} for distances of $n\times n/h$ many pairs (see $Q_2(n,n/h)$ term in \Cref{thm:intro:reduction}) 
to reconstruct an \emph{exact} $st$-shortest path. 
This is what allows our reduction to be very efficient, and for \Cref{thm:intro:directedintweighted} to be faster than previous work \cite{BergamaschiHGWW21}, and for \Cref{thm:intro:directedweighted} to match previous work of \cite{BrandN19}, despite returning the shortest path.

We give a brief description of the technical ideas used to prove \Cref{thm:intro:reduction}, a more detailed description is given in \Cref{sec:overview}.
As mentioned before, dynamic algebraic algorithms efficiently maintain $h$-bounded distances for small $h$.
Thus they are usually combined with graph techniques (e.g.~random hitting sets \cite{UllmanY91}) to decompose long paths $s\path t$ into segments $s\path h_1 \path h_2 \path ... \path t$ each of short length \cite{Sankowski05,BergamaschiHGWW21,BrandN19,BrandFN22}.
So the task to construct an $st$-shortest path reduces to the task of finding $h_ih_{i+1}$-shortest path for each $i$.

While techniques such as predecessor search can efficiently reconstruct the shortest path
for \emph{one} such segment,
repeating this for \emph{all} segments requires exact distances for upto $\tilde O(n\cdot n/h)$ vertex pairs. Existing algebraic data structures are too slow to compute the exact distances for so many pairs in subquadratic time. We can beat the quadratic barrier because our technique only requires $2$-approximate distances for $\tilde O(n\cdot n/h)$ vertex pairs.

This is done via
a pre-filtering step that for each segment reduces the search space of potential predecessors. 
This pre-filtering step uses an \emph{approximate} distance oracle, even if we later compute an \emph{exact} shortest path as in \Cref{thm:intro:directedintweighted}.
Using the approximate distances, we create for each segment $h_i\to h_{i+1}$ a set of plausible vertices $P_i \subset V$ that could potentially be on the $h_ih_{i+1}$-shortest path. 
We show that (i) the distances from $h_i$ to each $v\in P_i$ suffice to reconstruct the $h_ih_{i+1}$-shortest path,
(ii) while any one $P_i$ could be of size $O(n)$, the total size of all $P_i$ together is small: $\sum_i |P_i| = O(n \log nW)$ where $W$ is the largest edge weight in the graph.
This means we need only $O(n \log nW)$ exact $h$-bounded distance pairs, which algebraic data structures can compute in subquadratic time.

\paragraph{Other Related Work and Open Problems}

Besides shortest paths, there are also other dynamic problems that can only be maintained in subquadratic time when using algebraic techniques \cite{AbboudW14}, e.g.~maximum cardinality matching and reachability.
Algebraic data structures often have the downside that they maintain only some value (distance, size of the matching, whether there exists some path) but not the object itself (the path or matching), see e.g.~
\cite{KingS02,DemetrescuI00,Sankowski04,Sankowski05,Sankowski07,Sankowski08,BrandNS19,BrandN19,BrandS19,BrandFN22,GuR21,AlmanH22}.
For a long time, it was an open problem how to maintain the object itself. %
The work by \cite{KarczmarzMS22,BergamaschiHGWW21} and our results are the first results in that direction for reachability and shortest paths. However, for matching the question remains open, even against oblivious adversaries.
It is also open whether the shortest path (or reachability with path reporting) can be solved in the same complexity as maintaining the distance (or reachability information) or if there is some strict gap between maintaining the value vs. object. Conditional lower bounds \cite{AbboudW14,HenzingerKNS15} imply lower bounds for the value problem and no larger lower bounds are known for maintaining the object.

Similar open problems also occur in non-algebraic dynamic algorithms, i.e.~dynamic algorithms for approximate matching.
While the size of a matching can be approximated with better-than-2 approximation in $\polylog(n)$ update time \cite{Behnezhad22,BhattacharyaKSW22}, it is open whether the approximate matching itself can be computed in that time. 
Maintaining the matching itself for 2-approximations was already an active area of research by analyzing how to round fractional matchings to integral ones \cite{Wajc20,BhattacharyaK21}.

\subsection{Organization}

We start by giving some notation in the preliminaries (\Cref{sec:prelim}).
In \Cref{sec:overview} we then outline how to extend dynamic algorithms that maintain distances to support path queries.
These results are then formally proven in subsequent sections:
\Cref{sec:weighted} proves how to extend dynamic distance algorithms for weighted graphs also to maintain (approximately) shortest paths, and in \Cref{sec:fullalgo}, these techniques are combined with the dynamic distance algorithms from  \cite{Sankowski05,BrandN19} to obtain our results \Cref{thm:intro:directedintweighted,thm:intro:directedweighted}.
These results are randomized and work against an adaptive adversary.
At last, \Cref{sec:unweighted} considers unweighted undirected graphs and presents deterministic path reporting algorithms.

\subsection{Preliminaries}
\label{sec:prelim}

With high probability (w.h.p.) means with probability at least $1 - 1/n^c$
for any constant $c > 1$.

\paragraph{Graph notations.} \label{eq:sorted} We will use $\len_G(u, v)$ for the length of the edge $(u, v)$ in graph $G$ and $\dist_G(u, v)$ for distance between $u$ and $v$ in graph $G$. %
We write $\dist^h_G(u, v)$ for the $h$-bounded distances, i.e.~$\dist^h_G(u, v) = \dist_G(s,t)$ if $\dist_G(s,t) \le h$ and $\dist^h_G(u, v) = \infty$ otherwise.

We say $\tilde{d} \in [0,(1+\epsilon)h]\cup\{\infty\}$\footnote{We can w.l.o.g assume that any approximate data structure has output values within this range by setting all values larger than $(1+\epsilon)h$ to $\infty$.} is a $(1+\epsilon)$-approximate $h$-bounded $st$-distance, if (i)
$\dist_G(s,t)\le\tilde{d}$ and (ii) $\tilde{d}\le(1+\epsilon)\dist_G(s,t)$ when $\dist_G(s,t) \le h$. %
Similarly, we call any $\tilde d>0$ a $(1+\epsilon)$-approximate $h$-hop bounded $st$-distance if (i) $\tilde d \ge \dist_G(s,t)$, and (ii) $\tilde d = \dist_G(s,t)$ if there is a shortest path using at most $h$ hop.

We write $\pi^*_G(u, v)$ for a $uv$-shortest path in $G$.

We write $\mathcal{N}_G(v) = [u_1, u_2, \ldots ]$ for the neighbourhood of a vertex $v$ in an undirected graph $G$ sorted in ascending order by the edge length (breaking ties arbitrarily).
Similarly, $\mathcal{N}_{out}(v)$ and $\mathcal{N}_{in}(v)$ denote the out-adjacent and in-adjacent neighbourhoods of $v$ in directed graphs, sorted in ascending order by edge lengths.

\paragraph{Oracles.} 
Throughout, we will assume access to distance oracles. These oracles are given by running dynamic distance algorithms of previous work \cite{Sankowski05,BrandN19}.

\begin{definition}{\label{oracle-notation}}
    We write $\calO_G^{(1+\epsilon)}$ for an oracle that returns $(1+\epsilon)$-approximate $h$-bounded distance estimates on graph $G$ and supports the following operations 
    \begin{itemize}[nosep]
        \item $preprocess()$ -- preprocess graph $G$ in time $P_{1+\epsilon}$
        \item $update(e, w)$ -- update any edge $e$ in $G$ to weight $w$ in time $U_{1+\epsilon}$

        \item $queryAll(S_1, S_2)$ -- query all bounded distances between pairs $s_1 \in S_1, s_2 \in S_2$ in time $Q_{1+\epsilon}(|S_1|, |S_2|)$.

        \item $query(u, v)$ -- query bounded distance between $u$ and $v$ in time $Q_{1+\epsilon}(1, 1)$.
    \end{itemize}
    
    For example, in our pseudo-code, we will write $\calO_G^{(1+\epsilon)}.query(u,v)$ when querying the $(1+\epsilon)$-approximate $h$-bounded $uv$-distance in $G$. Also, for exact distance oracle we'll use notation $\calO_G$ for simplicity. %
\end{definition}

\paragraph{Integer weight rounding}

Our dynamic algorithm relies on the integer weight rounding technique by Zwick \cite{Zwick02}. This technique reduces the task of computing $(1+\epsilon)$-approximate distances on (positive) real weighted graphs to computing $(1+\epsilon)$-approximate distances on integer weighted graphs.
\begin{definition}
For real numbers $0 < A, B$ define the graph $G^{\prime}=\left(V, E^{\prime}\right)$ to be an $(A, B)$-rounded version of the graph with  edges
$E^{\prime}=\left\{(u, v) \in E \mid \len_G(u, v) \leq B\right\}$
\noindent
and integer edge weights 
$\len_{G'}(u, v)=\left\lceil A \len_G(u, v) / B\right\rceil$.
\end{definition}

\begin{lemma}[{\cite{Zwick02},\cite[Lemma 4.9]{BrandN19}}]
\label{AB-dist}

Let $G=(V, E)$ be a graph with $n$ nodes and real edge weights from $[1, W]$. For any $0<A, B$ let $G'$ be the $(A, B)$-rounded version of $G$.

Then for any path from $s$ to $t$ in $G$ of length $\dist_G(s, t) \leq B$ let $h$ be the number of its hops. We have $\operatorname{dist}_G(s, t) \leq(B / A) \operatorname{dist}_{G^{\prime}}(s, t) \leq \operatorname{dist}_G(s, t)+(B / A) h$.

\end{lemma} 

\begin{lemma}[{\cite{Zwick02},\cite[Lemma 4.10]{BrandN19}}]
\label{AB-min-dist}
Let $\epsilon \ge 0$, $0<\s<1$, $G=(V, E)$ be a graph with $n$ nodes and real edge weights from $[1, W]$. 
For $k = \left\lceil\log _2 n W\right\rceil$ and all $i=0,1,...,k$, define graph $G_i$ as $(A, B_i)$-rounded versions of $G$ where
$B_i=2^i, A=2 n^\s/ \epsilon$.

For the special case $\epsilon=0$ we let $k=0$, $G_0=G$, $A=B=Wn^\s$.

Then for any pair $s, t \in V$ we have $\dist_G(s, t) \leq \min _i\left(B_i / A\right) \dist_{G_i}(s, t)$ and if the $st$-shortest path uses at most $n^\s$ hops, then we also have
$
\min _i\left(B_i / A\right) \dist_{G_i}(s, t) \leq(1+\epsilon) \dist_G(s, t).
$
\end{lemma}
\noindent

\section{Technical Outline}
\label{sec:overview}

In this work, we present techniques for extending dynamic distance algorithms also to maintain an (approximately) shortest path against an adaptive adversary.
To outline our techniques, we will start with a simple warm-up on unweighted undirected graphs in \Cref{sec:overview:undirected}. We describe how to find the $st$-shortest path on this type of graph, which also serves as a demonstration of the issues that must be resolved when looking at weighted and/or directed graphs.
The subsequent \Cref{sec:overview:directed} then describes our techniques for directed and weighted graphs and how these techniques circumvent the problems from the previous subsection.

\subsection{Warm-up: Undirected Unweighted Graphs}
\label{sec:overview:undirected}

We start with a brief summary of how dynamic algebraic algorithms are used to maintain the distance. We then explain how to extend this result also to return the shortest path. 

\paragraph{Dynamic Distances}
Dynamic algebraic algorithms can efficiently maintain bounded distances.
For example, Sankowski \cite{Sankowski05} presented a dynamic algorithm that maintains $h$-bounded distances for any $h \ge 1$ on unweighted graphs in $O(hn^{1.529})$ update time, and $O(hn^{0.529})$ query time to return the $h$-bounded distance for any pair of vertices.
Via a common hitting set sampling argument (see e.g.~\cite{UllmanY91}), 
this result can be extended to unbounded distances: when sampling $\Theta(n/h)$ vertices $S\subset V$, 
any shortest path $s\path t$ is w.h.p.~split into segments $s\path h_1 \path h_2 \path ... \path t$ with $h_i \in S$ for all $i$, and each segment using at most $h$ edges. 
This leads to the following observation:
\begin{fact}\label{fact:hitting}
Let $R$ be a uniformly sampled random subset of $V$ of size $\Theta((n/h) \log n)$. 
Let graph $H = (V_H, E_H) = (R\cup\{s,t\}, E_H)$ be a complete graph where for all $(u,v)\in E_H$ the edge weight $\len_H(u,v)$ is the $h$-bounded $uv$-distance in $G$.
Then w.h.p.~$\dist_H(s,t) = \dist_G(s,t)$.
\end{fact}
So if we know the pairwise $h$-bounded distances $(R\cup\{s,t\})\times (R\cup\{s,t\})$, 
then we find the $st$-distance with just $O((n/h)^2)$ additional time for running Dijkstra's algorithm on graph $H$. 
The algorithm by Sankowski \cite{Sankowski05}  maintains these distances in $\tilde O(h(n/h)^2)$ additional time, resulting in overall $\tilde O(hn^{1.529}+n^2/h) = \tilde O(n^{1.765})$ time (where $h=n^{0.235}$) for maintaining the $st$-distance\footnote{%
This complexity assumes that set $S$ is fixed. If we later return the shortest path, then the adversary can learn set $S$. So to handle adaptive adversaries, we will resample set $S$ after each update which increases this complexity to $O(n^{\exact})$ as in \Cref{thm:intro:directedintweighted}.}.

\paragraph{Reconstructing the Path}

As outlined in the previous paragraph, we can assume that we already know that some $st$-shortest path consists of segments $s\path h_1 \path h_2 \path \ldots \path t$ where the $h_i$ (and their order) was found by running Dijkstra's algorithm on graph $H$.
For notational simplicity, we can add $s,t$ to set $S$ and define $h_0 = s$ and $h_k = t$ for some $k\in\N$, so the segments of the $st$-path are of form $h_i \path h_{i+1}$ for $i=0,...,k-1$.

To construct an $st$-shortest path in $G$, we can reconstruct shortest paths between $h_i$ and $h_{i+1}$ in $G$ for all $i$. 

Note the following observation for all $u,v\in V$: any vertex $w$ belongs to a $uv$-shortest path if and only if 
$\dist_G(u,w)+\dist_G(w,v) = \dist_G(u,v)$. Further, if $\dist_G(u,v)\le h$, then we can verify this property by querying the dynamic algebraic algorithm. 
Since each segment $h_i \path h_{i+1}$ is an $h_i h_{i+1}$-shortest path of length at most $h$,
this observation leads to an intuitive idea: 
let us run BFS from $h_i$ to $h_{i+1}$, but each vertex $w$ is only put in the BFS-queue, if 
\begin{align}\label{overview:check}
    \dist^{h}_G(v_{last},w)+\dist^{h}_G(w,h_{i+1}) = \dist^{h}_G(v_{last},h_{i+1}) 
\end{align}
where $v_{last}$ is the last recovered vertex on the $h_i h_{i+1}$-shortest path. 
Since querying these distances takes some time, we must bound how many queries we perform.

This question is answered by another important observation -- a ``no-shortcut'' argument,
which was also used in \cite{BernsteinC16,Bernstein17,GutenbergWW20}. 
For any $st$-shortest path and any $w\in V$, the vertex $w$ cannot be a neighbor of more than $3$ vertices on the shortest path. 
Otherwise we could construct a shortcut $s\path v_1 \to w \to v_4 \path t$, i.e.~an even shorter path:
Let $v_1,..,v_4$ be the vertices on the $st$-shortest path with neighbor $w$ (enumerated in the same order as on the path).
Then we have $$\dist(s,t) \le \dist(s,v_1) \underbrace{+~~2~~+}_{v_1\to w \to v_4} \dist(v_4,t) < \dist(s,v_1) + \underbrace{\dist(v_1,v_4)}_{\ge 3} + \dist(v_4,t) = \dist(s,t).$$
Thus we perform at most $O(n)$ distance queries since two distance queries are performed for each neighbor of a vertex on the shortest path and every vertex can be such a neighbor at most three times.

This directly implies that we can maintain the $st$-shortest path on unweighted undirected graphs in $O(hn^{1.529} + n^2/h + n\cdot h n^{0.529}) = O(n^{1.765})$ time, where the first two terms come from the dynamic algorithm by Sankowski \cite{Sankowski05} (outlined at the start of this subsection) and the last term comes from the $O(n)$ queries we perform to construct an $st$-shortest path.

\paragraph{Problems on weighted and/or directed graphs}
This idea is not immediately applicable to weighted or directed graphs. 
\begin{itemize}[nosep]
    \item If edges $(v_1, w), (w, v_4)$ have large weights, they don't necessarily give a shortcut. 
    \item If we only know approximate distances, then the above proof breaks down as we cannot verify if a vertex is on the shortest path. (Note that the exact distance cannot be maintained in subquadratic time on weighted graphs with polynomial edge weights under the APSP-conjecture \cite{AbboudW14}.)
    \item If the edges are directed, then we can no longer guarantee that $v_1 \rightarrow w \rightarrow v_4$ is a path as we only know $v_1 \rightarrow w \leftarrow v_4$. Thus we cannot bound the number of vertices that the BFS looks at via the no-shortcut argument.
    In particular, our algorithm might have to look at all $O(n^2)$ edges in the graph. 
    See \Cref{fig:overview-dijkstra-fails} for such an example.
\end{itemize}
The next subsection
presents a way to extend the no-shortcut argument to directed weighted graphs.

\subsection{Directed Weighted Graphs}
\label{sec:overview:directed}

\begin{figure} 
\begin{tikzpicture}
  [
      hset/.style={circle, fill=blue!50, minimum size=#1, inner sep=0pt, outer sep=0pt},
      hset/.default = 6pt,
      dot/.style = {circle, fill, minimum size=#1, inner sep=0pt, outer sep=0pt},
      dot/.default = 3pt,
      segment/.style={
        decorate, 
        decoration=snake,
        segment amplitude=.2mm,
        segment length=2.5mm,
        line after snake=10mm}, 
    segment2/.style={
        decorate, 
        decoration=snake,
        segment amplitude=.2mm,
        segment length=3.5mm,
        line after snake=0.5mm},
     back/.style={
         line width=0.5pt, color=orange, sloped, anchor=south, bend left=30}
    ],
    
    \def\lng{1}
    \def\lvl{0}
    \def\tip{Stealth[length=2mm, width=1mm]}
    \def \vertices{ 
        \node 
            [dot, label=above:$v_1$] 
            (r0)         
            at (0, \lvl) {}; 
        \node 
            [dot, label=above:$v_2$] 
            (r1)         
            at (2, \lvl) {}; 
        \node 
            []                       
            (r1_next)    
            at (2.5 , \lvl) {};
          
        \node 
            [] 
            (r2_prev) 
            at (3.5,  \lvl) {};
        \node 
            [hset, label=above:${v_t} (\color{blue}{h_i}\color{black})$] 
            (r2) 
            at (4,  \lvl) {};
        
        \node 
            [dot, label=above:${v_{t +1}}$] 
            (r3) 
            at (6, \lvl) {};
        
        \node 
            [] 
            (r3_next) 
            at (6 + 0.5, \lvl) {};
            
        \node 
            [dot, label=above:${v_{r -1}}$] 
            (r4) 
            at (9, \lvl) {};
            
        \node 
            []
            (r4_prev) 
            at (9 - 0.5, \lvl) {};
        
        \node 
            [hset, label=above:$v_{r}(\color{blue}{h_{i+1}}\color{black})$] 
            (r5) at    
            (11,    \lvl) {};
            
        \node 
            [] 
            (r5_next) 
            at    (11 + 0.5,    \lvl) {};

        \node (r6_prev) [] at (13 - 0.5, \lvl) {};
        \node (r6) [dot, label=above:$v_{n-2}$] at (13, \lvl) {};
        
        \node (r7) [dot, label=above:$v_n$] at (15, \lvl) {};
    }
    
    \vertices

    \draw [-{\tip}] (r2_prev.center) to (r2);
    \draw [-{\tip}] (r4_prev.center) to (r4);
    \draw [-{\tip}] (r6_prev.center) to (r6);
    \draw [-{\tip}] (r2) to (r3);
    \draw [-{\tip}] (r4) to (r5);
    
  \foreach \i in {0,...,7} {
    \foreach \j in {0,...,8} {
        \ifthenelse{\i > \j;}{\draw [-{\tip}, back] (r\i) to (r\j);}{}
        }
    };

    \draw [-{\tip} ] (r0) ->  (r1);
    \draw [- ] (r1) -> (r1_next.center);
    \draw [- ] (r3) -> (r3_next.center);
     \draw [- ] (r5) -> (r5_next.center);
     \draw [-{\tip} ] (r6) ->  (r7);
    
    \draw [-, dotted, anchor=north] 
        (r1_next.center) -- (r2_prev.center)
        (r3_next.center) to (r4_prev.center)
        (r5_next.center) -- (r6_prev.center)
        ;

\end{tikzpicture}
\vspace{-20pt}
\caption{
Directed unweighted graph $G$ with  $V = \{v_1, \ldots, v_n\}$ and $E = \{ (v_i, v_{i+1})\} \cup  \{(v_j, v_i) \mid  i<j \}$. Running BFS (even with truncation as in 
\eqref{overview:check}) 
to construct a path from left to right might iterate over all $O(n)$ backwards directed edges for each visited vertex until it finds the one edge going forward.}
\label{fig:overview-dijkstra-fails}
\end{figure}
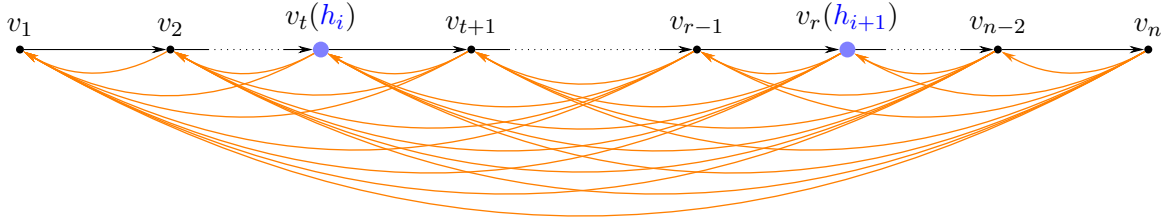

To best explain how to generalize the approach from the previous paragraph to directed graphs, we first give an alternative (slightly more complicated) argument for the undirected case, which is easier to generalize.

The shortest paths for any segments $h_i \path h_{i+1}$ and $h_j \path h_{j+1}$ such that $j - i \geq 4$ cannot share any adjacent vertices: 
\begin{align*}
     \underbrace{\left( \bigcup_{w \in \pi^*(h_i, h_{i+1})} \mathcal{N}(w) \right)}_{=:P_i} \cap  \underbrace{\left( \bigcup_{w \in \pi^*(h_j, h_{j+1})} \mathcal{N}(w) \right)}_{=:P_j}  = \emptyset
\end{align*}
Here the union $P_i$ on the left (or right $P_j$) are all vertices that are neighbors of an $h_ih_{i+1}$-shortest path (or $h_jh_{j+1}$-shortest path).
If $P_i$ and $P_j$ were to share a vertex $v$, i.e.~$(w_1, v), (v, w_2) \in E, w_1 \in \pi^*(h_i, h_{i+1}), w_2 \in \pi^*(h_j, h_{j+1})$, then there would be a shorter $st$-path via $v$:
\begin{align*}
    \dist(s, t) \leq  \dist(s, w_1) \underbrace{+~~2~~+}_{w_1 \to v \to w_2} \dist(w_2, t) < \dist(s, h_{i+1}) + \dist(h_{i+1}, h_j) + \dist(h_{j}, t) = \dist(s, t)
\end{align*}
Thus they cannot share a vertex. 

These unions $P_i, P_j$ can be seen as a set of ``plausible'' vertices: They are exactly the vertices for which our BFS search checks if they are on the shortest path, i.e.~for any vertex in $P_i$, it is plausible that they could be on a $h_ih_{i+1}$-shortest path.

To extend the approach to directed weighted graphs, we must find a better notion of plausible that allows for a similar ``no-shortcut'' argument.
The idea is to perform the search of vertices on a smaller set of ``plausible'' vertices for a given segment $h_i \path h_{i+1}$. Specifically, consider any set $P_i$ of vertices such that 
\begin{align*}
    \mathcal{N}_{out}^{(d)}(h_i) \subseteq P_i \subseteq \mathcal{N}_{out}^{(2d)}(h_i) \text{~~~~and~~~~}\mathcal{N}_{in}^{(d)}(h_{i+1}) \subseteq P_i \subseteq \mathcal{N}_{in}^{(2d)}(h_{i + 1})
\end{align*} 
where $d = \dist(h_i, h_{i+1})$. Here $\mathcal{N}_{out}^{(d)}(h_i)$ are all vertices reachable from $h_i$ with distance at most $d$. In particular, set $P_i$ contains all vertices $v\in V$ for which 
$\dist_G(h_i,v)\le d$ and $\dist_G(v,h_{i+1}) \le d$, 
and all vertices $v \in P_i$ satisfy $\dist_G(h_i,v) \le 2d$ and $\dist_G(v,h_{i+1}) \le 2d$.

We can consider this a set of ``plausible'' vertices because $P_i$ contains all vertices on any $h_ih_{i+1}$-shortest path.

Now let us extend the ``no-shortcut'' idea for these sets $P_i$'s to show that there are no (or not too many) intersections between $P_i$ and $P_j$ for $i\neq j$. 
For simplicity, let us assume that all segments $h_i \path h_{i+1}$ of the $st$-shortest path have roughly the same length despite the graph being weighted. 
(I.e.~$\dist_G(h_i,h_{i+1}) \in [2^{\alpha - 1}, 2^{\alpha})$ for all $i$. We will later argue why we can assume this.)
Then any $P_i$ and $P_j$ with $j - i \geq 7$ cannot share any vertices: 
\begin{align*}
     P_i \cap  P_j  = \emptyset
\end{align*}
If they were to share a vertex $v$, then there would be a shorter $st$-path via $v$ (\cref{fig:overview}):
\begin{align*}
    \dist(s, v) + \dist(v, t) &< \dist(s, h_i) + \underbrace{2 \cdot 2^\alpha + 2 \cdot 2^\alpha}_{h_i \to v \to h_{j+1}}
    +  \dist(h_{j+1}, t) 
     = \dist(s, h_{i}) + 2^{\alpha+2} + \dist(h_{j+1}, t) 
     \\
     &\leq \dist(s, h_{i}) + 2^{\alpha-1} (j + 1 - i) + \dist(h_{j+1}, t) 
     \\
     &\leq \dist(s, h_{i}) + \dist(h_{i}, h_{i+1}) + \ldots + \dist(h_{j}, h_{j+1}) + \dist(h_{j}, t)
\end{align*}
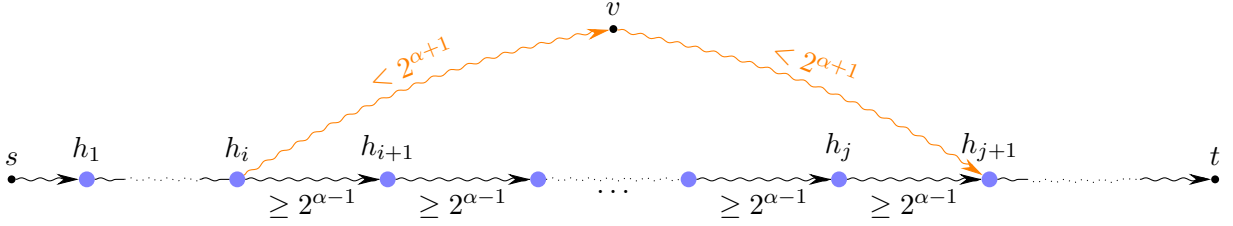
\begin{figure} 
\begin{tikzpicture}
  [
      hset/.style={circle, fill=blue!50, minimum size=#1, inner sep=0pt, outer sep=0pt},
      hset/.default = 6pt,
      dot/.style = {circle, fill, minimum size=#1, inner sep=0pt, outer sep=0pt},
      dot/.default = 3pt,
      segment/.style={
        decorate, 
        decoration=snake,
        segment amplitude=.2mm,
        segment length=2.5mm,
        line after snake=10mm}, 
    segment2/.style={
        decorate, 
        decoration=snake,
        segment amplitude=.2mm,
        segment length=3.5mm,
        line after snake=0.5mm},
     wcat/.style={
         anchor=north}
    ],
    
    \def\lng{1}
    \def\lvl{0}
    \def\tip{Stealth[length=3mm, width=1.5mm]}
    \def \vertices{ 
        \node 
            (s) 
            [dot, label=above:$s$] 
            at (0, \lvl) {}; 
        \node 
            [hset, label=above:$h_1$] 
            (h1) 
            at (\lng, \lvl) {}; 
        \node 
            [] 
            (h1_next) 
            at (\lng * 1.5 , \lvl) {};
          
        \node 
            [] 
            (hi_prev) 
            at (3 - \lng *0.5,   \lvl) {};
        \node 
            [hset, label=above:$h_i$] 
            (hi) 
            at (3, \lvl) {};
    
        \node 
            [hset, label=above:$h_{i + 1}$] 
            (hi1) 
            at (5, \lvl) {};
        \node 
            [] 
            (hi_next) 
            at (7 ,  \lvl) {};
        
        \node 
            [dot, label=above:$v$] 
            (v) 
            at (8, \lvl + 2) {};
        
        \node [hset] (sgm1s) at     (7, \lvl) {};
        \node [hset] (sgm1e) at     (9, \lvl) {};

        \node [hset, label=above:$h_{j}$] (hj) at         (11,            \lvl) {};
        \node [hset, label=above:$h_{j + 1}$] (hj1) at    (13, \lvl) {};
        \node [] (hj1_next) at                             (13.5,   \lvl) {};

        \node [] (t_prev) at (15, \lvl) {};
        \node [dot, label=above:$t$] (t) at (16, \lvl) {};
        \node [] (d1) at    (5.5,     \lvl - 1.25) {};
        \node [] (d2) at    (12.5,    \lvl - 1.25) {};
    }
    
    \vertices
    
    \draw [-{\tip}, segment, color=orange, sloped, anchor=south, bend left=10] (hi) to node  {$< 2^{\alpha + 1}$}  (v);
    \draw [-{\tip}, segment, color=orange, sloped, anchor=south, bend left=10] (v)  to node {$< 2^{\alpha + 1}$} (hj1);

    \draw [-{\tip}, segment, wcat] 
    (hi)    -> node {$\geq 2^{\alpha - 1}$}       (hi1);
    \draw [-{\tip}, segment, wcat] 
    (hi1)    -> node {$\geq 2^{\alpha - 1}$}     (sgm1s);
    \draw [-, segment, dotted, wcat] 
    (sgm1s) -> node {$\ldots$}   (sgm1e);
    \draw [-{\tip}, segment, wcat] 
    (sgm1e) -> node {$\geq 2^{\alpha - 1}$}   (hj);
    \draw [-{\tip}, segment, wcat] 
    (hj)    -> node {$\geq 2^{\alpha - 1}$}     (hj1);

    \draw [-{\tip},segment] (s) ->  (h1);
    \draw [-,segment] (h1) -> (h1_next.center);
    \draw [-,segment] (hi_prev.center) -> (hi);
    \draw [-,segment] (hj1) -> (hj1_next.center);
    \draw [-{\tip},segment] (t_prev.center) -> (t);
        
    \draw [-, dotted, segment, anchor=north] 
        (h1_next.center) -- (hi_prev.center)
        (hj1_next.center) -- (t_prev.center) 
        ;
        
\end{tikzpicture}
\vspace{-30pt}
\caption{A possibility of a shortcut in graph $G$ between $h_i$ and $h_{j+1}$.}
\label{fig:overview}
\end{figure}
Hence, in total we iterate over at most $O(n)$ vertices if we iterate over all $P_i$ for all $i$.
We can reconstruct any $h_ih_{i+1}$-shortest path by iterating over each $P_i$ as follows: 
Sort $P_i$ based on their distance to $h_i$ and then iterate over $P_i$ to always find the next successor on the $h_ih_{i+1}$-shortest path via a distance comparison as in \eqref{overview:check}.\footnote{Here, we assume that we can compute the exact distance. This is true for small integer weighted graphs. We will later discuss how to handle real weighted graphs for which no exact distance can be maintained under the APSP conjecture \cite{AbboudW14}.}
By $\sum_i |P_i| = O(n)$ we need to query only $O(n)$ distances to reconstruct the $st$-shortest path.

The assumption on all $\dist(h_i, h_{i+1}) \in [2^{\alpha-1}, 2^\alpha)$ can be generalized by splitting segments in groups $S_{\alpha} = \{(h_i, h_{i+1}) \mid  \dist(h_i, h_{i+1}) \in [2^{\alpha-1}, 2^\alpha)\}$ for $\alpha = 1, \ldots \lceil \log nW \rceil$ and applying the ``no-shortcut'' argument to each group. 
This way we get $$
\sum_i P_i = \sum_\alpha \sum_{P_i \text{ belongs to }S_\alpha} |P_i| = \sum_\alpha O(n) = O(n \log (nW)). 
$$
So we increase the number of distance queries by at most an $O(\log nW)$ factor.

To complete the argument, we must construct the sets $P_i$ for all $i$.
This can be done by querying $2$-approximate $h$-hop bounded distances for all pairs in $S\times V$ and $V\times S$ (remember, $S = \{h_i \mid i\}$ where $h_i$'s lie on the shortest path).

Querying large batches of approximate pairwise distances can be done much more efficiently than querying exact distances for individual pairs, see e.g.~\cite{BrandN19} (\Cref{san-bn}). 
So the construction of the $P_i$ is only a small cost of our algorithm.

To summarize, the complexity of reconstructing an $st$-shortest path on directed weighted graphs is given by the following
\begin{itemize}[nosep]
    \item Query exact $h$-hop bounded pairwise distances between hitting set vertices (i.e.~pairs in $S\times S$) to construct graph $H$. 
    \item Run Dijkstra on $H$ to obtain segments $s=h_0\path h_1 \path ... \path h_k = t$ in $\tilde O(n^2/h^2)$ time.
    \item Query $2$-approximate $h$-hop bounded distances for pairs $S\times V$ and $V\times S$ to construct plausible sets of vertices $P_i$.
    \item Reconstruct an $h_ih_{i+1}$-shortest path for each $i$ by iterating over $P_i$. This needs $O(n\log nW)$ exact $h$-hop bounded distance queries.
\end{itemize}
This leads to the following result:
\simplePathConstruction*
Observe that our path reconstruction just needs access to a dynamic algorithm that can maintain approximate and exact $Wh$-bounded distance oracles. 
The complexities of our dynamic algorithms stated in \Cref{sec:intro:results} are obtained by using the dynamic distance algorithms by \cite{Sankowski05,BrandN19}.
If in the future faster dynamic distance algorithms are constructed, then our path reporting data structures become faster as well. 

\paragraph{Real Weighted Graphs}

So far, we assumed that we have access to an exact distance oracle to reconstruct the shortest path for each $h_i\path h_{i+1}$ segment.
For real weighted graphs, however, one cannot maintain exact distances in subquadratic time under the APSP conjecture \cite{AbboudW14}.
To still be able to reconstruct an approximate shortest path on real weighted graphs we use the integer weight rounding technique from \cite{Zwick02}.
Note that by \Cref{AB-min-dist} we can compute approximate distances on a real weighted graph $G$ by computing distances on graphs $G_j$ for $j=1, \ldots ,\log_2 \lceil nW \rceil$, each with small integer weights.
In particular, for any segment $h_i \path h_{i+1}$ there is some $j$ where the exact $h_ih_{i+1}$-shortest path on graph $G_j(h_i,h_{i+1})$ corresponds to an approximate $h_ih_{i+1}$-shortest path on $G$.
The graph $G_j(h_i,h_{i+1})$ has small integer weights so we can compute the exact distances on this graph.
So we can reconstruct for any segment $h_i\path h_{i+1}$ the exact shortest path on $G_j(h_i,h_{i+1})$ and thus an approximate $h_ih_{i+1}$-shortest path on $G$.

\newcommand{\thmstatement}[4]{
    Suppose, there exist #3-bounded distance oracles #1  with the corresponding time-complexities (as in \Cref{oracle-notation}) for any integer weighted #2 graph $X$ with $|V| = n$.  Then there exists an algorithm that supports the following operations on a #2 graph $G = (V, E, w), |V| = n$ with real weights from $[1, W]$.
}
\newcommand{\thmps}{
    The dynamic algorithm is randomized and correct w.h.p.~with one-sided error and works against an adaptive adversary.
}

\newcommand{\thmprep}{
    $preprocess()$ -- preprocesses $G$ in time
}

\newcommand{\thmupd}{
    $update(e, w)$ -- updates an edge $e$ with weight $w$ in time
}

\newcommand{\thmquery}[1]{
    $query(s, t)$ -- for any pair of vertices $s, t$ returns #1 path in  time
}

\section{Path Reporting on Weighted Graphs}
\label{sec:weighted}

In this section we prove combinatorial blackbox reductions that allow us to efficiently reconstruct an $st$-shortest path in graph $G$ when given access to a distance oracle for $G$.
We use these reductions in \Cref{sec:fullalgo} together with the dynamic distance algorithms by \cite{Sankowski05,BrandN19} to maintain the $st$-shortest path against an adaptive adversary.

As an example, here we prove reductions such as \Cref{thm:aprx-dir-gen} which construct approximate shortest paths in directed graphs. 
Later results \Cref{thm:aprx-undir-gen,thm:exact-dir-gen} work for exact shortest paths, or undirected graphs.

\begin{restatable}[directed, approximate]{theorem}{blackbox}
    \label{thm:aprx-dir-gen}
    \thmstatement{$ \mathcal{O}_{X}, \mathcal{O}^{(1 + \epsilon)}_{X}$ and $\mathcal{O}^{(2)}_{X}$}{directed}{$4 n^\s / \epsilon$}{real}
     
    \begin{itemize}[nosep]  
        \item \thmprep $\tilde O\left(\log W \cdot (P_{2} + P_{1+\epsilon} + P)\right)$
        \item \thmupd $$
        \tilde O\left(\log W \cdot 
        \left(
            U_{2} + U_{1+\epsilon} + U
            +
            Q_{1+\epsilon}(n^{1-\s}, n^{1-\s}) 
            + 
            Q_{2}(n^{1-\s}, n) 
        \right)
        \right)
        $$
        
        \item  \thmquery{an $(1+\epsilon)$-approximate $st$-shortest}
        $$
        \tilde O\left( \log W \cdot 
            \right(
                Q_{1+\epsilon}(1, n^{1-\s}) 
                + 
                Q_{2}(1, n) 
                +
                n^{2 - \s} 
                +
                n \cdot Q(1, 1)
            \left)
        \right)
        $$
    \end{itemize}
    \thmps
\end{restatable}

The organization of this section is as follows.
We first define certain auxiliary graphs, used by our reduction, in \Cref{sec:weighted:auxiliary}.
The definition and notation defined there are used throughout this section.
In particular, it defines random graphs based on hitting set arguments that are used to split any $st$-shortest path into shorter segments $s=h_1\path h_2 \path ... h_k = t$.
Then in \Cref{sec:segments:directed}, we describe how to find the $h_i h_{i+1}$-shortest path for any one such segment, when the graph is directed.
\Cref{sec:segments:undirected} does the same, but for undirected graphs.
At last, \Cref{sec:blackbox} combines these tools to prove \Cref{thm:aprx-dir-gen} and its variants.

\subsection{Auxiliary graphs}
\label{sec:weighted:auxiliary}

Throughout this section, we assume $G=(V,E)$ is the original input graph.
Let $\epsilon \ge 0$ be an accuracy parameter and $0<\s<1$ be a hop-parameter, where $n^\s$ will be used for our hop bounds.
We let $G_x$ be the $(A,B_x)$-rounded version of $G$, as in \Cref{AB-min-dist}, for $x=0,1,...,O(\log(nW))$.

Our reduction constructs an auxiliary graph $H$ on $\tilde{O}(n^{1-\s})$ vertices $V_H \subset V$ with the property $\dist_H(u,v) \approx \dist_G(u,v)$ for all $u,v\in V_H$.
The exact definition of $H$ is given in \Cref{construct-h}.

\begin{definition}\label{construct-h}
    Given graph $G=(V,E)$, accuracy parameter $\epsilon\ge0$, hop parameter $0<\s<1$,
    for $x=0,1,...,O(\log(nW))$ let $G_x$ be the $(A,B_x)$-rounded graphs as in \Cref{AB-min-dist}.
    Let $R\subset V_H \subset V$ where $R$ is a uniformly at random sampled set of $\tilde{\Theta}(n^{1-\s})$ vertices.
    Given $(1+\epsilon)$-approximate 
    $A$-bounded distance estimates $\Delta_x \in \R^{V_H\times V_H}$ on each $G_x$ for the pairs $V_H \times V_H$, define $H$ as follows: 
    
    $H=(V_H,V_H \times V_H)$ with edge weights\footnote{For simplicity, assume we remove all edges with $\len_H(u,v) = \infty$ so all edge weights are finite.} $\len_H(u,v) = \min_x B_x/A \cdot \Delta_x(u,v)$ for each $u,v\in V_H$.

    Note that when $\epsilon=0$ and $G$ has integer weights from $[1, W]$, there is only one copy $G_0$ of $G$ and the edge weights in the corresponding $H$ are exact $n^\s$-bounded distances in $G$ between any $h_i, h_j \in V_H$. 
\end{definition}
The following \Cref{hi-hj-aprx} states that graph $H$ indeed approximates the distances in $G$.
Since we use common techniques such as hitting-sets and integer weight rounding to construct $H$, 
we will defer the proof to \Cref{appendix}.

\begin{restatable}{lemma}{hihjaprx}\label{hi-hj-aprx}
    For any $u,v\in V_H$,
    w.h.p.~$   \dist_G(u, v) \le \dist_H(u, v) \leq (1 + O(\epsilon) ) \dist_G(u, v) $.\\
    Equality holds in case of $\epsilon=0$.
\end{restatable}  

During the construction of $H$ as in \Cref{construct-h}, we can store which distance estimate $\Delta_x(u,v)$ was used for any $u,v\in V_H$. This implies an assignment of some $(A,B_x)$-rounded graph $G_x$ to each edge.
\begin{definition}\label{well-aprx}
For $H$ as in \Cref{construct-h}, let $G_1,G_2,...$ be the respective $(A,B_x)$-rounded graphs.
For each edge $(u, v)$ of $H$ let
\begin{align*}
    x^* &= \arg \min_{x} \frac{B_x}{A} \Delta_x(u,v). %
\end{align*}
We define unique $G_*(u, v) := G_{x^*}$ (breaking ties arbitrarily) and say that it \textbf{well-approximates} the edge $(u,v)$. %
\end{definition}

As outlined in \Cref{sec:overview}, our path reconstruction is based on ``no-shortcut'' arguments.
To bound the complexity, we must bound how often we look at any vertex. We will argue that if we look at a vertex too often, then there must have been a shortcut contradicting the shortest path.
For this type of argument we need lower bounds on the distances in $G$.
To derive these, we need the following definition.

\begin{definition}\label{def:weight_category}
Given a $st$-shortest path $\pi_G=(s=h_1,h_2,...,h_k=t)$ in $H$, we group the edges $(h_i,h_{i+1})$ for $i=1,...,k-1$ into the following weight-categories $S_\alpha$ for $0\le\alpha \le \log nW$.
\begin{align*}
    S_\alpha = \left[ \sigma^{(\alpha)}_1, \sigma^{(\alpha)}_2, \ldots \sigma^{(\alpha)}_l \right] = \left[ (h_i, h_{i+1}) \mid \len_H(h_i, h_{i+1}) \in [2^{\alpha-1}, 2^{\alpha}) \right]
\end{align*}

The order of edges in the categories  matches the order in the $st$-shortest path: 

if edge $\sigma^{(\alpha)}_a = (h_i, h_{i+1})$, $\sigma^{(\alpha)}_b = (h_j, h_{j+1})$ then $i < j$ if and only if $a < b$.
\end{definition}

We can now state the required lower bounds on $\dist_G(u,v)$.

\begin{lemma} \label{hw-lemma}
Given an $st$-shortest path in $H$ and the split of its edges into weight categories as in \Cref{def:weight_category},
let $\sigma^{(\alpha)}_a = (h_i, h_{i + 1})$ and 
$\sigma^{(\alpha)}_b = (h_j, h_{j+1})$  ($a < b$) 
be edges from the same weight category $S_\alpha$.
Then the following holds:
    \begin{align*}
        \dist_G(h_{i}, h_{j + 1}) \geq 2^{\alpha-2}(b - a)
    \end{align*}
\end{lemma}

\begin{proof}
    The shortest path in $H$ from $h_i$ to $h_j$ contains all edges $\sigma^{(\alpha)}_{a}, \sigma^{(\alpha)}_{a+1}, \ldots , \sigma^{(\alpha)}_{b - 1}$ of the length at least $2^{\alpha - 1}$. Hence, $\dist_H(h_i, h_{j + 1}) \geq 2^{\alpha-1}(b - a)$. 
    
    Combining the inequality above with \Cref{hi-hj-aprx} we get a needed result:
    \begin{align*}
         (1 + \epsilon ) \dist_G(h_i, h_{j+1}) &\geq \dist_H(h_i, h_{j+1}) \geq 2^{\alpha-1}(b - a)\Rightarrow \\
         \dist_G(h_i, h_{j+1}) &\geq 2^{\alpha-2}(b - a)
    \end{align*}
\end{proof}

\subsection{Reconstructing Path Segments on Directed Graphs}%
\label{sec:segments:directed}

Given $H$ as in \Cref{construct-h}, and a shortest path $(h_1,h_2, \ldots ,h_k)$ in $H$, our goal 
is to recover an (approximately) $h_1h_k$-shortest path in $G$.
Note that we have $\dist_G(h_1,h_k) \le (1+\epsilon)\dist_H(h_1,h_k)$ (\Cref{hi-hj-aprx}), 
and that the edge weight $(h_i,h_{i+1})$ in $H$ corresponds to the (approximate) length of an $h_ih_{i+1}$-path in some $(A,B_x)$-rounded graph $G_x$ by \Cref{construct-h}. 
So for each $i=1,2,3...$, our task is to reconstruct the $h_i h_{i+1}$-shortest path in the graph $G_x$ that well approximates edge $(h_i,h_{i+1})$. 
To bound the complexity when constructing an $h_ih_{i+1}$-shortest path in $G_x$, we restrict our search space onto a smaller set of ``plausible'' vertices.

\begin{definition}\label{def:directed_plausible}
Given an edge $\sigma = (h_i, h_{i+1})$ in $H$ with finite length, well approximated by $G_x := G_*(h_i, h_{i+1})$ (\Cref{well-aprx}), $2A$-bounded $\mathcal{O}^{(2)}_{G_x}$ and $\mathcal{O}_{G_x}$.
We define the set of \textbf{plausible} vertices for edge $(h_i, h_{i+1})$: 
\begin{align*}
    P_{i} = \left\{ v \in V \mid  \mathcal{O}^{(2)}_{G_x}.query(h_i, v),  \mathcal{O}^{(2)}_{G_x}.query(v, h_{i+1}) \le  2 \cdot \mathcal{O}_{G_x}(h_i, h_{i+1}) \right\}
\end{align*}
Note that $\mathcal{O}_{G_x}.query(h_i, h_{i+1})<\infty$. This is because $\len_H(h_i, h_{i+1})<\infty$ and this length came from some $(1+\epsilon)$-approximate $A$-bounded distance estimate on $G_x$, so $\dist_{G_x}(h_i,h_{i+1}) \le (1+\epsilon)A < 2A$.
In particular, $\mathcal{O}^{(2)}_{G_x}.query(h_i, v)$, $\mathcal{O}^{(2)}_{G_x}.query(v, h_{i+1}) <\infty$ are for all $v\in P_i$.
\end{definition}

Our main result of this subsection is the following \Cref{thm:ssp-dir}, which states that we can reconstruct the $h_ih_{i+1}$-shortest path in $G_x$, when given the set of plausible vertices $P_i$. We later prove in \Cref{not-too-many-total} that $|P_i|$ is small on average, which then implies that our algorithm is efficient.

\begin{theorem} \label{thm:ssp-dir} 
Given an edge $(h_i, h_{i+1})$, the $(A, B_x)$-rounded version $G_x$ well-approximating the edge, the set of plausible vertices $P_{i}$ for the edge, and $\mathcal{O}_{G_x}$ is an exact $2A$-bounded distance oracle for $G_x$.

Then \Cref{alg:ssp-dir} recovers an $(h_i, h_{i+1})$-shortest path in $G_x$ in $O(|P_{i}|)$ calls to $\mathcal{O}_{G_x}.query(*, *)$.
\end{theorem}

\begin{algorithm2e}[t!] 
\caption{Reporting a shortest $h_i, h_{i+1}$ path in $(A,B_x)$-rounded $G_x = G_*(h_i, h_{i+1})$ for the directed case} 
\label{alg:ssp-dir}
\SetKwProg{Proc}{procedure}{}{}
\Proc{$\textsc{ShortestSubpath}(G_x, (h_i, h_{i+1}), P_{i})$}{
Let $\mathcal{O}_{G_x}$ be an exact $2A$-bounded oracle. \\
(By definition of $P_i$, $\dist_{G_x}(h_i, v)$ and $\dist_{G_x}(v, h_{i+1}) \le 4A$ for each $v\in P_i$) \\
${queue} \leftarrow $ sorted $v \in  P_{i}$ by $\mathcal{O}_{G_x}.query(h_i, v)$ in ascending order \label{algo:line:queue}\\
$v_{last} \gets h_i$ \\
$\pi \gets [~]$\\
\For{$v$ in $queue$ \label{alg:line:while}}{
    \uIf{$\len_{G_x}(v_{last}, v) + \mathcal{O}_{G_x}.query(v, h_{i+1}) = \mathcal{O}_{G_x}.query(v_{last}, h_{i+1})$ \label{algo:line:check}}{
        $\pi \gets \pi || v$ \label{algo:line:add}\\
        $v_{last} \gets v$
    }
    \Else{continue}
    }
\Return $\pi$
}
\end{algorithm2e}

To prove correctness of \Cref{thm:ssp-dir}, we must first prove that we can indeed restrict the search space onto the set of plausible vertices $P_i$.

\begin{lemma} 
    \label{on-sp-are-plausible} If $G_x$ well-approximates the edge $(h_i, h_{i+1})$ then all the vertices from the $h_i h_{i+1}$-shortest path in $G_x$ are plausible, i.e.~an element of $P_i$.
\end{lemma}

\begin{proof}

    Suppose $G_x$ well-approximates edge $(h_i, h_{i+1})$. 
    As the edge $(h_i, h_{i+1})$ has finite length in $H$ we know that $\dist_{G_x}(h_i, h_{i+1}) \leq (1+\epsilon)A$, hence, $\mathcal{O}_{G_x}.query(h_i, h_{i+1}) = \dist_{G_x}(h_i, h_{i+1})$.
    
    If $v$ is on any shortest path in $G_x$ from $h_i$ to $h_{i+1}$ then: 
    \begin{align*}
        \dist_{G_x}(h_i, v) + \dist_{G_x}(v, h_{i+1}) &= \dist_{G_x}(h_i, h_{i+1} ) 
         &\Rightarrow\\
        \mathcal{O}^{(2)}_{G_x}.query(h_i, v) \text{ and }  \mathcal{O}^{(2)}_{G_x}.query(v, h_{i+1}) &\leq 2 \cdot \mathcal{O}_{G_x}.query(h_i, h_{i+1}) 
    \end{align*}
    Note that both queries on the LHS don't return $\infty$ because $$\dist_{G_x}(h_i, v), \dist_{G_x}(v, h_{i+1}) \leq \dist_{G_x}(h_i, h_{i+1} )  \leq (1+\epsilon)A.$$
\end{proof}

We can now prove \Cref{thm:ssp-dir}, which states that \Cref{alg:ssp-dir} indeed reconstructs an $h_i h_{i+1}$-shortest path on $G_x$.

\begin{proof}[Proof of \Cref{thm:ssp-dir}]

Consider an execution of \Cref{alg:ssp-dir}, where we are given an edge $(h_i,h_{i+1})$ from $H$, 
and an integer rounded graph $G_x$ that well approximates this edge, 
and a set of plausible vertices $P_i$.

\paragraph{Correctness}

First, we remark that any calls to $\calO_{G_x}(v,h_{i+1})$ and $\calO_{G_x}(h_i,v)$ in \Cref{alg:ssp-dir} 
never return $\infty$ for vertices $v$ on the $h_ih_{i+1}$-shortest path.
Since $\len_H(h_i,h_{i+1})$ has finite value, we know $\dist_{G_x}(h_i,h_{i+1}) \le (1+\epsilon)A < 2A$, so all vertices $v$ on the $h_i h_{i+1}$-shortest path have small enough distance for a $2A$-bounded oracle $\calO_{G_x}$ to return correct finite distances.

Next we argue that we indeed construct the shortest path.
First note that by \Cref{on-sp-are-plausible} all vertices of any $h_i h_{i+1}$-shortest path in $G_x$ are present in the queue defined in \cref{algo:line:queue}.

Let us show by induction that  after every added vertex $v$ to the path $\pi$ on \cref{algo:line:add} (\Cref{alg:ssp-dir}) there exist some $h_i h_{i+1}$-shortest path on $G_x$ that starts with $\pi$. The base case is obvious as we start with $\pi = (h_i)$.

Suppose, we already constructed path $(h_i, v_1, \ldots, v_g) = h_i \path v_g$. 
By assumption, there exists a $h_i h_{i+1}$-shortest path with $h_i \path v_g$ as a head. 
Note that any vertex $w$, for which there exists a $h_i h_{i+1}$-shortest path starting with $h_i \path v_g \rightarrow w$, is in the queue because $\dist_{G_x}(h_i, w) > \dist_{G_x}(h_i, v_g)$ as the queue is ordered. 
Also any vertex $w \in \mathcal{N}_{G_x}(v_g)$ is on some $v_g h_{i+1}$-shortest path if and only if $\dist_{G_x}(v_g, h_{i+1}) = \len_{G_x}(v_g, w) + \dist_{G_x}(w, h_{i+1})$ that is being checked on \cref{algo:line:check} of the algorithm. 
Hence, the next vertex we append to our path will be on some $v_g h_{i+1}$-shortest path and correspondingly on some $h_i h_{i+1}$-shortest path. 

\paragraph{Complexity} 
Executing \cref{algo:line:queue} takes $O(|P_{i}|)$ calls to $\mathcal{O}_{G_x}.query(*, *)$ and additional time $O(|P_{i}| \cdot \log |P_{i}|) = \tilde O(|P_{i}|)$ time for sorting.

As \cref{alg:line:while} is executed  $O( |P_{i}|)$ times and every loop requires $O(1)$ calls to $\mathcal{O}_{G_x}.query(*, *)$  the total time is dominated by $\tilde O(|P_{i}|)$ calls to $\mathcal{O}_{G_x}.query(*, *)$.

\end{proof}

The purpose of the set $P_i$ of plausible vertices is to restrict the search space and thus result in a faster algorithm.
To give good complexity bounds, we must prove that the set $P_i$ is small on average.

Recall that set $P_i$ is defined w.r.t.~an edge $(h_i,h_{i+1})$ that represents some segment on some (approximately) $st$-shortest path $s=h_1 \path h_2 \path ... h_k=t$ on $G$.
Since our final aim is to reconstruct the entire (approximately) $st$-shortest path, we will reconstruct the $h_i h_{i+1}$-shortest paths for all $i=1,...,k-1$ by repeatedly calling \Cref{alg:ssp-dir} (\Cref{thm:ssp-dir}).
The following \Cref{not-too-many} and \Cref{not-too-many-total} bound the total size of all $P_i$ that we construct for $i=1,...,k-1$.

\begin{lemma}\label{not-too-many} There are $O(n)$ plausible vertices across all 
segments from the same weight category (\Cref{def:weight_category}) $S_\alpha = \{ \sigma \mid \len_H(\sigma) \in [2^{\alpha - 1}, 2^\alpha)\}$:

\begin{align*}
    \sum_{
     \substack
     {
        i: \\
        (h_i, h_{i+1}) \in S_\alpha
     }
    } |P_{i}| =  O(n)
\end{align*}
\end{lemma}

\begin{proof}

 Suppose a vertex $v$ is plausible for edges $\sigma^{(\alpha)}_a = (h_i, h_{i+1})$ and $\sigma^{(\alpha)}_b = (h_j, h_{j+1})$, $a < b$ and is well-approximated by $G_x$ (\Cref{fig:dir}): 
  \begin{align*}
    \dist_{G_x}(h_i, v) &\leq \mathcal{O}^{(2)}_{G_x}.query(h_i, v)  \leq 2 \cdot \mathcal{O}_{G_x}.query(h_i, h_{i+1})  \leq 2 \cdot \mathcal{O}_{G_x}^{(1 + \epsilon)}.query(h_i, h_{i+1}) =\\
    &= 2 \frac{A}{B_x}\len_H(h_i, h_{i+1}) =  2 \frac{A}{B_x}\dist_H(h_i, h_{i+1}) \leq \frac{A}{B_x} \cdot 2^{\alpha + 1}
    \end{align*}
Here we used $\len_H(h_i, h_{i+1}) =  \dist_H(h_i, h_{i+1})$, which follows from the fact that $(h_i,h_{i+1})$ was on some $h_1h_k$-shortest path.
Hence, by \Cref{AB-dist}

$$\dist_{G}(h_i, v)  \leq \frac{B_x}{A}\dist_{G_x}(h_i, v) \leq 2^{\alpha + 1}$$

Similarly, we have that $ \dist_{G}(v, h_{j+1}) \leq  2^{\alpha + 1}$.
Combining both inequalities:
\begin{align*}
    \dist_{G}(h_i, h_{j+1}) \leq  2^{\alpha + 2}
\end{align*}

On the other hand, by \Cref{hw-lemma}
\begin{align*}
    \dist_{G}(h_i, h_{j+1}) \geq  2^{\alpha - 1}(b - a)
\end{align*}

Therefore, $b - a = O(1)$. That means every vertex can be plausible only for $O(1)$ edges from the same $S_\alpha$. So in total
\begin{align*}
    \sum_{
     \substack
     {
        i: \\
        (h_i, h_{i+1}) \in S_\alpha
     }
    } |P_{i}| =  O(n)
\end{align*}
\end{proof}

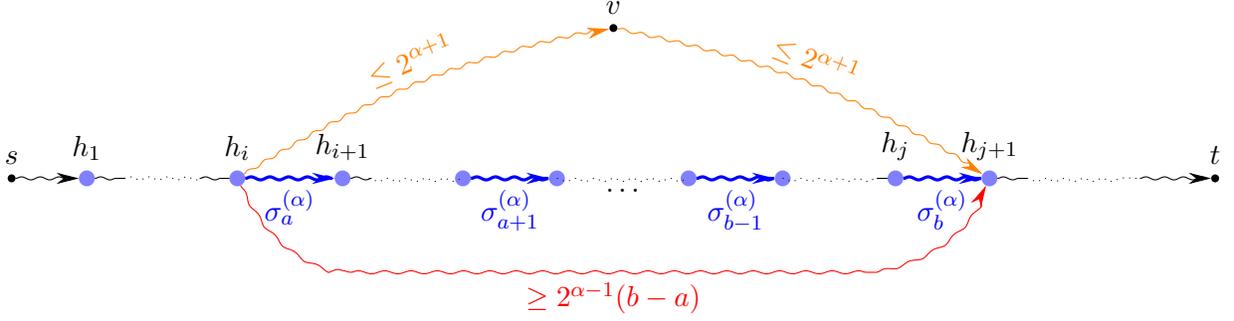
\begin{figure} 
\begin{tikzpicture}
  [
      hset/.style={circle, fill=blue!50, minimum size=#1, inner sep=0pt, outer sep=0pt},
      hset/.default = 6pt,
      dot/.style = {circle, fill, minimum size=#1, inner sep=0pt, outer sep=0pt},
      dot/.default = 3pt,
      segment/.style={
        decorate, 
        decoration=snake,
        segment amplitude=.2mm,
        segment length=2.5mm,
        line after snake=10mm}, 
    segment2/.style={
        decorate, 
        decoration=snake,
        segment amplitude=.2mm,
        segment length=3.5mm,
        line after snake=0.5mm},
     wcat/.style={
         blue, line width=1.2pt, anchor=north}
    ],
    
    \def\lng{1}
    \def\lvl{0}
    \def\tip{Stealth[length=3mm, width=1.5mm]}
    \def \vertices{ 
        \node 
            [dot, label=above:$s$] 
            (s) 
            at (0, \lvl) {}; 
        \node 
            [hset, label=above:$h_1$] 
            (h1) 
            at (\lng, \lvl) {}; 
        \node 
            [] 
            (h1_next) 
            at (\lng * 1.5 , \lvl) {};
          
        \node 
            [] 
            (hi_prev) 
            at (3 - \lng *0.5, \lvl) {};
        \node 
            [hset, label=above:$h_i$] 
            (hi) 
            at (3, \lvl) {};
    
        \node 
            [hset, label=above:$h_{i + 1}$] 
            (hi1) 
            at (3 + 1.4*\lng ,    \lvl) {};
        \node 
            [] 
            (hi_next) 
            at (3 + 1.75*\lng ,  \lvl) {};
        
        \node 
            [dot, label=above:$v$] 
            (v) 
            at (8, \lvl + 2) {};
        
        \node 
            [hset] 
            (sgm1s) 
            at (6, \lvl) {};
            
        \node 
            [hset] 
            (sgm1e) 
            at (6  + 1.25*\lng, \lvl) {};
        
        \node 
            [hset] 
            (sgm2s) 
            at (9, \lvl) {};
        \node 
            [hset] 
            (sgm2e) 
            at (9  + 1.25*\lng, \lvl) {};
        
        \node 
            [] 
            (hj_prev) 
            at (12 - \lng *0.5,\lvl) {};
        \node 
            [hset, label=above:$h_{j}$] 
            (hj) 
            at (13 - 1.25*\lng, \lvl) {};
        \node 
            [hset, label=above:$h_{j + 1}$] 
            (hj1) 
            at  (13, \lvl) {};
        \node 
            (t) 
            [dot, label=above:$t$] 
            at (16, \lvl) {};
        \node 
            [] 
            (d1) 
            at (3 + 1.5*\lng,     \lvl - 1.25) {};
        \node 
            [] 
            (d2) 
            at (13 - 1.5*\lng,    \lvl - 1.25) {};
        \node 
            [] 
            (hj1_next) 
            at  (13.5,   \lvl) {};

        \node 
            [] 
            (t_prev) 
            at (15, \lvl) {};
    }
    
    \vertices
    
    \draw [-{\tip}, segment, color=orange, sloped, anchor=south, bend left=10] (hi) to node  {$\leq 2^{\alpha + 1}$}  (v);
    \draw [-{\tip}, segment, color=orange, sloped, anchor=south, bend left=10] (v)  to node {$\leq 2^{\alpha + 1}$} (hj1);

    \draw [-{\tip}, segment, wcat] 
    (hi)    -> node {$\sigma^{(\alpha)}_a$}       (hi1);
    \draw [-{\tip}, segment, wcat] 
    (sgm1s) -> node {$\sigma^{(\alpha)}_{a+1}$}   (sgm1e);
    \draw [-{\tip}, segment, wcat] 
    (sgm2s) -> node {$\sigma^{(\alpha)}_{b-1}$}   (sgm2e);
    \draw [-{\tip}, segment, wcat] 
    (hj)    -> node {$\sigma^{(\alpha)}_{b}$}     (hj1);

    \draw [-{\tip},segment] (s) ->  (h1);
    \draw [-,segment] 
        (h1) -> (h1_next.center);
    \draw [-,segment] 
        (hi_prev.center) -> (hi);
    \draw [-,segment] 
        (hi1) -> (hi_next.center) 
        (hj_prev.center)  -> (hj);
    \draw [-,segment] 
        (hj1) -> (hj1_next.center);
    \draw [-{\tip},segment] 
        (t_prev.center) -> (t);

    \draw [-, dotted, segment, anchor=north] 
        (h1_next.center) -- (hi_prev.center)
        (hi_next.center) -- (sgm1s.center) 
        (sgm1e.center)  -- node {$\ldots$} (sgm2s.center) 
        (sgm2e.center)  -- (hj_prev.center)
        (hj1_next.center) -- (t_prev.center); 
    
    \draw [-, segment2, anchor=north, bend right=25, color=red] (hi) to (d1.center);
    \draw [-{\tip}, segment2, anchor=north, bend right=25, color=red] (d2.center) to (hj1) ;
    \draw [-, segment2, anchor=north, color=red] (d1.center) to node  {$\geq 2^{\alpha - 1}(b - a)$}  (d2.center);
\end{tikzpicture}
\caption{A possibility of a shortcut in graph $G$ between $h_i$ and $h_{j+1}$. Red curly arrow indicates a $h_ih_{j+1}$-shortest path in $G$.}
\label{fig:dir}
\end{figure}

The following \Cref{not-too-many-total} bounds the total size of all plausible sets $P_i$ that we have when reconstructing an approximately $st$-shortest path $s=h_1 \path h_2 \path ... \path h_k = t$, by reconstructing each segment $h_i \path h_{i+1}$ via \Cref{alg:ssp-dir} (\Cref{thm:ssp-dir}).

\begin{corollary}\label{not-too-many-total} There are $O(n \log (W n))$ plausible vertices across all
segments:
\begin{align*}
    \sum_{i = 1}^{k} |P_{i}| =  O(n \log (nW))
\end{align*}
\end{corollary}

\begin{proof}
\begin{align*}
    \sum_{i} |P_{i}| = 
    \sum_{\alpha = 1}^{\log_2 \lceil nW \rceil}
    \sum_{
     \substack
     {
        i: \\
         (h_i, h_{i+1}) \in S_\alpha
     }
    } |P_{i}|
    =
    O(n \log (nW))
\end{align*}
\end{proof}

\subsection{Reconstructing Path Segments on Undirected graphs} %
\label{sec:segments:undirected}

\begin{algorithm2e}[t!] 
\caption{Reporting a shortest $h_i, h_{i+1}$ path in $(A,B_x)$-rounded $G_x = G_*(h_i, h_{i+1})$ for the undirected case} 
\label{alg:ssp-undir}
\SetKwProg{Proc}{procedure}{}{}
\Proc{$\textsc{ShortestSubpath}(G_x, (h_i, h_{i+1}))$}{
    Let $\mathcal{O}_{G_x}$ be an exact $2A$-bounded oracle. \\
    $v_{last} \gets h_i$\\
    $\pi  \gets [h_i]$\\
    set of vertices seen before: $Q \gets \{ \}$ \\
    \While{$v_{last}$ is not $h_{i+1}$ \label{algo:undir:line:for}}
    {
        \For{$w$ in $\mathcal{N}_{G_x}(v_{last})$ \label{algo:undir:n} \tcp{
        it is sorted according to \cref{eq:sorted}
        }}{
            \uIf{$w \in Q$}{
             continue \label{line:continue}
            }
            $Q.add(w)$\\
            \uIf{$\len_{G_x}(v_{last}, w) + \mathcal{O}_{G_x}.query(w, h_{i+1}) = \mathcal{O}_{G_x}.query(v_{last}, h_{i+1})$ \label{algo:undir:line:check}}{
             $\pi \gets \pi  || w$\\
             $v_{last} \gets w$\\ 
             break \label{line:break}\\
            }
        }
    }
}
\end{algorithm2e}

Given $H$ as in \Cref{construct-h}, and a shortest path $(s=h_1,h_2, \ldots ,h_k=t)$ in $H$, our goal 
is to recover an (approximately) $h_1h_k$-shortest paths in $G$.
Note that we have $\dist_G(h_1,h_k) \le (1+\epsilon)\dist_H(h_1,h_k)$ (\Cref{hi-hj-aprx}), 
and that the edge weight $(h_i,h_{i+1})$ in $H$ corresponds to the (approximate) length of an $h_ih_{i+1}$-path in some $(A,B_x)$-rounded graph $G_x$ by \Cref{construct-h} that well approximates edge $(h_i,h_{i+1})$. 
So for each $i=1,2,3,...$, our task is to reconstruct the $h_i h_{i+1}$-shortest path in the respective graph $G_x$.

For undirected graphs, we use the same approach as in the directed case that we outlined in \Cref{sec:segments:directed}.
The main difference is that we do not need to perform a pre-filtering of the vertices onto a smaller set of plausible vertices. 
For the directed case in \Cref{sec:segments:directed}, we had to compute a certain set of plausible vertices (\Cref{def:directed_plausible}).
However, in the undirected case here, it suffices to define ``plausible vertices'' only for the sake of analysis. 
We do not need to compute this set.

\begin{definition}\label{def:undirected:plausible}
Given an edge $\sigma = (h_i, h_{i+1})$ we define a set of \textbf{plausible} vertices for edge $(h_i, h_{i+1})$ using uniquely-defined $G_*(h_i, h_{i+1}) = G_x$ that well-approximates it:

\begin{align*}
    P_{i} = \left\{ v \in V \mid  \dist_{G_x}(h_i, v) \leq   \dist_{G_x}(h_i, h_{i+1}) \right\}
\end{align*}
\end{definition}

Our main result of this subsection is the following \Cref{thm:ssp-undir}, which states that \Cref{alg:ssp-undir} correctly reconstructs the $h_ih_{i+1}$-shortest path in $G_x$.
Notably, the complexity scales in the size of set $P_i$. 
The later \Cref{not-too-many-total:undir} shows that when reconstructing the $h_ih_{i+1}$-shortest path for each $i$, the total sum of all $|P_i|$ is nearly linear, giving a very efficient bound on the complexity of \Cref{alg:ssp-undir}.

\begin{theorem} \label{thm:ssp-undir}  
Given undirected $G$, an edge $(h_i, h_{i+1})$, the $(A, B)$-rounded version $G_x$ well-approximating the edge, and $\mathcal{O}_{G_x}$ is the exact $2A$-bounded distance oracle for $G_x$.
Then \Cref{alg:ssp-undir} recovers a $(h_i, h_{i+1})$-shortest path in $G_x$.

The time complexity is $O(n \log A)$ 
plus the time to perform $O(|P_{i}|)$ calls to $\mathcal{O}_{G_x}.query(*, *)$, where $P_{i}$ is a set of plausible vertices for $(h_i, h_{i+1})$ (\Cref{def:undirected:plausible}).
\end{theorem}
\begin{proof}%

We prove the correctness by induction over the number of iterations in \Cref{algo:undir:line:for} of \Cref{alg:ssp-undir}.
Suppose we already reconstructed a path $h_i \path v_{last}$ such that there exists a $h_i h_{i+1}$-shortest path with $h_i \path v_{last}$ as a head. 
To prove that this path is correctly extended, we first need to argue that list $Q$ only contains vertices that are plausible, i.e.~for all $u\in Q$ we have $\dist_{G_x}(h_i,u) \le \dist_{G_x}(h_i,h_{i+1})$.

\paragraph{$Q$ stored plausible vertices}
Vertex $u$ was added to $Q$ because it is the neighbor of some $v_j$ on the $h_iv_{last}$-shortest path constructed by the algorithm, i.e.~$h_i \path v_j \to v_{j+1} \path v_{last}$. 
    Observe that $u$ appeared in $\mathcal{N}_{G_x}(v_j)$ earlier than $v_{j+1}$ because no other neighbor of $v_j$ is iterated over in \Cref{algo:undir:n} after $v_{j+1}$, because of the \textit{break} in \Cref{line:break}.
    Since we iterate over the neighbors in ascending order of their edge weights, we have $\len_{G_x}(v_j, u) \le \len_{G_x}(v_j, v_{j+1})$.
    Using the fact that $h_i \path v_j$ is the head of some $h_i v_{last}$-shortest path we get:
    \begin{align}
        \dist_{G_x}(h_{i}, u) \leq&~
        \dist_{G_x}(h_{i}, v_{j}) + \len_{G_x}(v_j, u) \leq  \dist_{G_x}(h_{i}, v_{j}) + \len_{G_x}(v_j, v_{j+1}) = \dist_{G_x}(h_{i}, v_{j + 1})\notag \\
        \le&~\dist_{G_x}(h_{i}, v_{last}) \le \dist_{G_x}(h_{i}, h_{i+1}) \label{eq:contradiction}
    \end{align}
    Thus vertex $u$ must be plausible.

\paragraph{Correctness of the head}
First let us remark, the calls to $\calO_{G_x}.query$ never returns $\infty$ in \Cref{alg:ssp-undir} for vertices on the shortest path.
Since $\len_H(h_i,h_{i+1})$ has finite value, we know $\dist_{G_x}(h_i,h_{i+1}) \le (1+\epsilon)A \le 2A$, so all vertices on the shortest path have distance at most $2A$ to $h_i$ and $h_{i+1}$, so a $2A$-bounded oracle $\calO_{G_x}$ suffices.

We now argue that \Cref{alg:ssp-undir} correctly appends some neighbor $w$ of $v_{last}$ to the path $h_i\path v_{last} \to w$, such that this path is the head of some $h_i h_{i+1}$-shortest path. 

The algorithm iterates over the neighbors of $v_{last}$.
If we append some neighbors $w$ to the path $h_i \path v_{last} \to w$, 
then we have by \Cref{algo:undir:line:check} that
$$\len_{G_x}(v_{last},w) + \dist_{G_x}(w, h_{i+1}) = \dist_{G_x}(v_{last}, h_{i+1}).$$
Thus $h_i \path v_{last} \to w$ must be the head of a $h_i h_{i+1}$-shortest path.

Next, we must argue that we do in-fact always append some neighbor of $v_{last}$ to the path.
Since $h_i\path v_{last}$ is the head of a $h_i h_{i+1}$-shortest path, there must be a neighbor $w$ of $v_{last}$ where $h_i \path v_{last} \to w \path h_{i+1}$ must be a $h_i h_{i+1}$-shortest path.
If no neighbor of $v_{last}$ is appended to $h_i\path v_{last}$, then that must mean $w$ was in $Q$, and it was not appended because of \Cref{line:continue}.
However, $w$ cannot be in $Q$ as otherwise we would have
\begin{align*}
    \dist_{G_x}(h_{i}, w) \leq&~\dist_{G_x}(h_{i}, v_{last}) ~~~~~~\text{by \eqref{eq:contradiction}} \\
    <&~ \dist_{G_x}(h_{i}, v_{last}) + \len_{G_x}(v_{last}, w) = \dist_{G_x}(h_{i}, w).
\end{align*}
which is a contradiction.

\paragraph{Complexity}

Note that we perform exactly one call to $\mathcal{O}_{G_x}.query(*, *)$ for each vertex in $Q$.
As argued before, these are all plausible vertices, so we can bound the number of oracle calls by $|P_i|$.

Next, we must bound how many vertices we iterate through in \Cref{algo:undir:n}.
Here iterating over $w\in\mathcal{N}_{G_x}(v_{last})$ can also interpreted as iterating over edges $\{v_{last},w\}$ incident to $v_{last}$.
To bound over how many edges we iterate, let us split all edges in $G_x$ into groups based on their edge weights, i.e. group $E_\ell$ are edges with weight in $[2^\ell, 2^{\ell+1})$.
Graph $G_x$ has edge weights in $[1, A]$ so there are $O(\log A)$ such groups.
We will argue that we iterate over at most $O(n)$ edges in each $E_\ell$, thus in total we iterate over at most $O(n \log A)$ vertices in \Cref{algo:undir:n}.

Assume vertex $w$ was iterated over 4 separate times in \Cref{algo:undir:n} for the same weight class $E_y$.
Let $v_1,v_2,v_3,v_4,v_5$ be the respective vertices on the $h_i h_{i+1}$-shortest path (not necessarily consecutive) for which we iterated over $w$.
We know by the neighbors being iterated over in ascending order of the edge weights that $\len_{G_x}(v_k,w) \le \dist_{G_x}(v_k,v_{k+1})$ for $k=1,...,5$, because the first edge on the $v_kv_{k+1}$-shortest path must have had edge weight at least $\ln_{G_x}(v_k,w)$.
By $\len_{G_x}(v_k,w) \in [2^y, 2^{y+1})$ for some $y$ (since they are from the same weight class) we have
\begin{align*}
\dist_{G_x}(v_1, v_5) 
\le &~
\len_{G_x}(v_1,w) + \len_{G_x}(w,v_4) 
\le 
\dist_{G_x}(v_1, v_2) + \len_{G_x}(w,v_4) \\
< &~
\dist_{G_x}(v_1, v_2) + 2^{y+1} 
= 
\dist_{G_x}(v_1, v_2) + 2^{y} + 2^y \\
\le &~
\dist_{G_x}(v_1, v_2) + \len_{G_x}(v_2,w) + \dist_{G_x}(v_3, w) \\
\le&~
\dist_{G_x}(v_1, v_2) + \dist_{G_x}(v_2,v_3) + \dist_{G_x}(v_3, v_4) 
= 
\dist_{G_x}(v_1, v_5)
\end{align*}
which is a contradiction.
So any vertex $w$ can be iterated over at most $4$ times for the same weight class.
Thus in total, \Cref{algo:undir:n} looks at at most $O(n \log A)$ vertices.

\end{proof}

The next \Cref{not-too-many:undir} and \Cref{not-too-many-total:undir} bound how many plausible vertices can exist. 
This can be used to bound the total time complexity of applying \Cref{alg:ssp-undir} (\Cref{thm:ssp-undir}) to all segments of some $st$-shortest path.

\begin{lemma}\label{not-too-many:undir} There are $O(n)$ plausible vertices across all 
segments from the same weight category $S_\alpha = \{ \sigma \mid \len_H(\sigma) \in [2^{\alpha - 1}, 2^\alpha)\}$:

\begin{align*}
    \sum_{
     \substack
     {
        i: \\
        (h_i, h_{i+1}) \in S_\alpha
     }
    } |P_{i}| =  O(n)
\end{align*}
\end{lemma}

\begin{proof}

 Suppose a vertex $v$ is plausible for edges $\sigma^{(\alpha)}_a = (h_i, h_{i+1})$ and $\sigma^{(\alpha)}_b = (h_j, h_{j+1})$, $a < b$ and is well-approximated by $G_x$ (\Cref{fig:undir}): 
 Just as in the proof of \Cref{not-too-many} we have that 
 $$\dist_{G}(h_i, v)  
 \leq \dist_{G_x}(h_i, v)
 \leq \dist_{G_x}(h_i, h_{i+1})
 \leq (1 + \epsilon) \dist_{G}(h_i, h_{i+1})
 \leq  2^{\alpha + 1}$$
 
 and similarly $\dist_{G}(h_j, v)  \leq  2^{\alpha + 1}$.

As this is an undirected graph, we can combine both inequalities:
\begin{align}
    \dist_{G}(h_i, h_{j+1}) \leq  \dist_{G}(h_i, v) + \dist_{G}(v, h_{j}) + \dist_{G}(h_{j}, h_{j+1}) \leq 2^{\alpha + 1} + 2^{\alpha + 1} + 2^\alpha \leq 2^{\alpha + 3}
\end{align}

On the other hand, by \Cref{hw-lemma}
\begin{align*}
    \dist_{G}(h_i, h_{j+1}) \geq  2^{\alpha - 1}(b - a)
\end{align*}

Therefore, $b - a = O(1)$. That means every vertex can be plausible only for $O(1)$ edges from the same $S_\alpha$. So in total
\begin{align*}
    \sum_{
     \substack
     {
        i: \\
        (h_i, h_{i+1}) \in S_\alpha
     }
    } |P_{i}| =  O(n)
\end{align*}

\end{proof}

\begin{figure} 
\begin{tikzpicture}
  [
      hset/.style={circle, fill=blue!50, minimum size=#1, inner sep=0pt, outer sep=0pt},
      hset/.default = 6pt,
      dot/.style = {circle, fill, minimum size=#1, inner sep=0pt, outer sep=0pt},
      dot/.default = 3pt,
      segment/.style={
        decorate, 
        decoration=snake,
        segment amplitude=.2mm,
        segment length=2.5mm,
        line after snake=0.00mm}, 
    segment2/.style={
        decorate, 
        decoration=snake,
        segment amplitude=.2mm,
        segment length=3.5mm,
        line after snake=0.0mm}
    ],
    
    \def\lng{1}
    \def\lvl{0}
       \def \vertices{ 
        \node 
            [dot, label=above:$s$] 
            (s) 
            at (0, \lvl) {}; 
        \node 
            [hset, label=above:$h_1$] 
            (h1) 
            at (\lng, \lvl) {}; 
        \node 
            [] 
            (h1_next) 
            at (\lng * 1.5 , \lvl) {};
          
        \node 
            [] 
            (hi_prev) 
            at (3 - \lng *0.5, \lvl) {};
        \node 
            [hset, label=above:$h_i$] 
            (hi) 
            at (3, \lvl) {};
    
        \node 
            [hset, label=above:$h_{i + 1}$] 
            (hi1) 
            at (3 + 1.4*\lng ,    \lvl) {};
        \node 
            [] 
            (hi_next) 
            at (3 + 1.75*\lng ,  \lvl) {};
        
        \node 
            [dot, label=above:$v$] 
            (v) 
            at (8 - 1.25*0.5*\lng, \lvl + 2) {};
        
        \node 
            [hset] 
            (sgm1s) 
            at (6, \lvl) {};
            
        \node 
            [hset] 
            (sgm1e) 
            at (6  + 1.25*\lng, \lvl) {};
        
        \node 
            [hset] 
            (sgm2s) 
            at (9, \lvl) {};
        \node 
            [hset] 
            (sgm2e) 
            at (9  + 1.25*\lng, \lvl) {};
        
        \node 
            [] 
            (hj_prev) 
            at (12 - \lng *0.5,\lvl) {};
        \node 
            [hset, label=above:$h_{j}$] 
            (hj) 
            at (13 - 1.25*\lng, \lvl) {};
        \node 
            [hset, label=above:$h_{j + 1}$] 
            (hj1) 
            at  (13, \lvl) {};
        \node 
            (t) 
            [dot, label=above:$t$] 
            at (16, \lvl) {};
        \node 
            [] 
            (d1) 
            at (3 + 1.5*\lng,     \lvl - 1.25) {};
        \node 
            [] 
            (d2) 
            at (13 - 1.5*\lng,    \lvl - 1.25) {};
        \node 
            [] 
            (hj1_next) 
            at  (13.5,   \lvl) {};

        \node 
            [] 
            (t_prev) 
            at (15, \lvl) {};
    }
    
    \vertices
    
    \draw [-, segment, color=orange, sloped, anchor=south, bend left=10] (hi) to node  {$\leq 2^{\alpha + 1}$}  (v) to node {$\leq 2^{\alpha + 1}$} (hj);

    \draw [-,segment, blue, line width=1.2pt, anchor=north] 
    (hi)    -- node {$\sigma^{(\alpha)}_a$}       (hi1) 
    (sgm1s) -- node {$\sigma^{(\alpha)}_{a+1}$}   (sgm1e) 
    (sgm2s) -- node {$\sigma^{(\alpha)}_{b-1}$}   (sgm2e) 
    (hj)    -- node {$\sigma^{(\alpha)}_{b}$}     (hj1);

    \draw [-,segment] 
        (s) --  (h1);
     \draw [-,segment] 
        (h1) -- (h1_next.center) 
        (hi_prev.center) -- (hi)   
        (hi1) -- (hi_next.center) 
        (hj_prev.center)  -- (hj) 
        (hj1) -- (hj1_next.center)
        (t_prev.center) -- (t);
    
    \draw [-, dotted, segment, anchor=north] 
        (h1_next.center) -- (hi_prev.center);
    \draw [-, dotted, segment, anchor=north]
        (hi_next.center) -- (sgm1s.center) ;
    \draw [-, dotted, segment, anchor=north]
        (sgm1e.center) -- node {$\ldots$} (sgm2s.center) ;
    \draw [-, dotted, segment, anchor=north]
        (sgm2e.center) -- (hj_prev.center);
    \draw [-, dotted, segment, anchor=north] 
        (hj1_next.center) -- (t_prev.center);
        
    \draw [-, segment2, anchor=north, bend right=25, color=red] (hi) to (d1.center);
    \draw [-, segment2, anchor=north, bend right=25, color=red](d2.center) to (hj1) ;
    \draw [-, segment2, anchor=north, color=red] (d1.center) to node  {$\geq 2^{\alpha - 1} (b-a)$}  (d2.center);
\end{tikzpicture}
\caption{A possibility of a shortcut in graph $G$ between $h_i$ and $h_{j+1}$. Red curly arrow indicates a shortest $h_ih_{j+1}$ path in $G$.}
\label{fig:undir}
\end{figure}
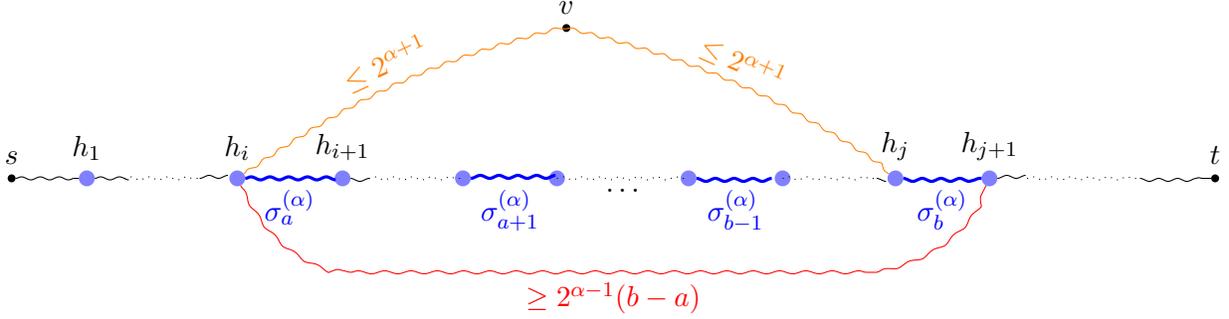

\begin{corollary}\label{not-too-many-total:undir} There are $O(n \log (W n))$ plausible vertices across all
segments:
\begin{align*}
    \sum_{i = 1}^{k} |P_{i}| =  O(n \log (nW))
\end{align*}
\end{corollary}

\subsection{Blackbox Reductions}
\label{sec:blackbox}

We now have all tools available to prove the blackbox reduction from dynamic shortest path to dynamic distance algorithms.
The main idea is to use hitting set arguments to split any $st$-shortest path into shorter segments, then use the results from \Cref{sec:segments:directed,sec:segments:undirected} to reconstruct the path for each such segment.

We start by giving our approximate result for directed weighted graphs. %
\blackbox*

\begin{algorithm2e}[th] 
\caption{Querying an $(1+\epsilon)$-approximate $st$-shortest path in weighted $G$} 
\label{wda-final-query}
\SetKwProg{Proc}{procedure}{}{}
\tcp{Blue lines are executed on directed graphs only.}
\Proc{$\textsc{Query}(G, s, t, \epsilon)$}{
Add vertices $s,t$ to $H$ and add edges incident to $s$ and $t$ as in \Cref{construct-h} \label{alg:q:line:new-h-edges}\\
For this, we must call $\calO_{G_x}^{(1+\epsilon)}.query(v,V_H)$ and $\calO_{G_x}^{(1+\epsilon)}.query(V_H,v)$ for $v\in\{s,t\}$ for all $x$.\\
$(s=h_0, h_1, \ldots, h_{k-1}, h_k=t) \gets \textsc{Dijkstra}\left(H, (s, t)\right)$ \label{alg:q:line:dijkstra}\\
\color{blue}  Compute and save $\mathcal{O}^{(2)}_{G_x}.queryAll(\{s\}, V)$ and $\mathcal{O}^{(2)}_{G_x}.queryAll(V, \{t\})$ for all $x$\color{black} \label{alg:q:line:o2query}\\
\color{blue}  Construct $P_0, \ldots, P_{k}$ using results of $\mathcal{O}^{(2)}_{G_x}.queryAll(V_H, V)$, $\mathcal{O}^{(2)}_{G_x}.queryAll(V, V_H)$, $\mathcal{O}^{(2)}_{G_x}.queryAll(V, \{s\})$ and $\mathcal{O}^{(2)}_{G_x}.queryAll(\{t\}, V)$\color{black}  \label{alg:q:line:plausible}\\
\For{$i = 0, \ldots, k$ \label{alg:q:line:g-copies} }{
        $G_x \gets G_*(h_i, h_{i+1})$ \\
        $\pi_{i} \gets \textsc{ShortestSubpath}(G_x, (h_i, h_{i+1}), \color{blue} P_i \color{black})$ \label{alg:q:line:ssp} \tcp*[h]{\Cref{alg:ssp-dir,alg:ssp-undir}, depending on if $G$ is (un-)directed.}\\
    }
\Return $\pi$
}
\end{algorithm2e}

\begin{algorithm2e}[th] 
\caption{Updating the data structure for given new weight $c$ for an edge $e$} 
\label{wda-final-update}
\SetKwProg{Proc}{procedure}{}{}
\tcp{Blue lines are executed on directed graphs only.}
\Proc{$\textsc{Update}(G, \epsilon, e, c)$}{
\For{all maintained $G_x$ (as defined in \Cref{AB-min-dist}) \label{alg:upd:line:g-copies}}{
    \uIf{ $c \leq B_x$} 
        {   
            $c_x \gets \left\lceil A c / B_x\right\rceil$ \\
            \color{blue}  $\mathcal{O}^{(2)}_{G_x}.update(e, c_x)$ \color{black} \label{alg:upd:line:o2}\\
            $\mathcal{O}^{(1 + \epsilon)}_{G_x}.update(e, c_x)  $ \label{alg:upd:line:o1e}\\
            $\mathcal{O}_{G_x}.update(e, c_x)$ \label{alg:upd:line:o}\\
        }
    }
    Construct $H$ as in \cref{construct-h} \label{alg:upd:line:h}\\
    \color{blue}  Compute and save results of $\mathcal{O}^{(2)}_{G_x}.queryAll(V_H, V), \mathcal{O}^{(2)}_{G_x}.queryAll(V, V_H)$ \color{black} \label{alg:upd:line:o2qall}\\
}
\end{algorithm2e}

\begin{proof}%

Consider the algorithm defined by \Cref{wda-final-update} for update operation and \Cref{wda-final-query} for query operation (all lines are executed).

\paragraph{Correctness.} The fact that concatenated shortest subpaths form a $(1 + O(\epsilon))$-approximate $st$-shortest path in $G$ follows from \Cref{hi-hj-aprx} and the correctness of each subpath comes from \Cref{thm:ssp-dir}.

\paragraph{Complexity.} As we have $O(\log (nW) ) $ copies $G_x$, the total query time (the lines are from \Cref{wda-final-query}) is 
 \begin{align*}
     \tilde{O}\left(  
        \underbrace{
        Q_{1+\epsilon}(1, n^{1-\s}) \cdot %
         \log W }
         _
         {\text{Edges $V_H \times \{s, t \}$, (\Cref{alg:q:line:new-h-edges})}}
         +
         \underbrace{n^{2 - 2\s} }
         _
         {\substack{\text{Dijkstra} \\ \text{(\Cref{alg:q:line:dijkstra})}}}
         + 
         \underbrace{
         Q_2(n, 1) \cdot %
          \log W }
         _
         {\text{\Cref{alg:q:line:o2query}} }
         + 
         \underbrace{n^{2 - \s} \log W }
         _
         {\text{$P_i$'s, \Cref{alg:q:line:plausible}} }
         + 
         \underbrace{
         n \cdot Q(1, 1) \cdot %
         \log W}
         _
         {\text{all ShortestPaths (\Cref{alg:q:line:g-copies})}}
     \right).
 \end{align*}
 Note that the result of $queryAll(V,V_H)$ for \Cref{alg:q:line:plausible} was precomputed during the \textsc{Update} routine (\Cref{wda-final-update}), so the cost of that query does not occur here for \textsc{Query}.
  The complexity of \Cref{alg:q:line:g-copies} is bounded by 
 \Cref{thm:ssp-dir} and \Cref{not-too-many-total}.
 
 Total update time (the lines are from \Cref{wda-final-update}): 
 \begin{align*}
     \tilde{O}
     \left( 
      ( 
          \underbrace{
           U_{2} + U_{1+\epsilon} + U
          }
          _
          {
          \text{(\Cref{alg:upd:line:o2}, 
          \Cref{alg:upd:line:o1e},
          \Cref{alg:upd:line:o})
          }}   
      +
      \underbrace{
        Q_{1+\epsilon}(n^{1-\s}, n^{1-\s})
        }
      _
      {\text{edges of }H, \text{(\Cref{alg:upd:line:h})}}
      +
      \underbrace{
      Q_2(n, n^{1-\s})
      }
      _
      {\mathcal{O}^{(2)}_{G_x}.queryAll(V, V_H), \text{(\Cref{alg:upd:line:o2qall}})}  
      )
     \log W
     \right) 
 \end{align*}

\end{proof}

\begin{theorem}[undirected, approximate]
    \label{thm:aprx-undir-gen} 

    \thmstatement{$\mathcal{O}_{X}$ and $\mathcal{O}^{(1 + \epsilon)}_{X}$}{undirected}{$4 n^\s / \epsilon$}{real}
     
    \begin{itemize}[nosep]
        \item \thmprep $\tilde O\left(\log W \cdot (P_{1+\epsilon} + P)\right)$
        \item \thmupd $$
        \tilde O\left(\log W \cdot 
        \left(
            U_{1+\epsilon} + U
            +
            Q_{1+\epsilon}(n^{1-\s}, n^{1-\s}) 
        \right)
        \right)
        $$
        
        \item  \thmquery{an $(1+\epsilon)$-approximate $st$-shortest}
        $$
        \tilde O\left( \log W \cdot 
            \right(
                Q_{1+\epsilon}(1, n^{1-\s}) + n^{2 - 2\s} + n \cdot Q(1, 1)
            \left)
        \right)
        $$
    \end{itemize}
    \thmps
\end{theorem}

\begin{proof}

Consider the algorithm defined by \Cref{wda-final-update} for update operation and \Cref{wda-final-query} for query operation (only black lines are executed, for $\textsc{ShortestSubpath}$ use \Cref{alg:ssp-undir} that doesn't need $P_i$ as input).

\paragraph{Correctness.} The fact that concatenated shortest subpaths form a $(1 + O(\epsilon))$-approximate $st$-shortest path in $G$ follows from \Cref{hi-hj-aprx}.
The correctness of each subpath comes from \Cref{thm:ssp-undir}.

\paragraph{Complexity.} As we have $O(\log (nW) ) $ copies $G_x$, each of which has edge weights at most $A=O(n^\s/\epsilon)$, the total query time (the lines are from \Cref{wda-final-query}) is: 
 \begin{align*}
     \tilde{O}\left(  
        \underbrace{
        Q_{1+\epsilon}(1, n^{1-\s}) \cdot %
         \log W }
         _
         {\text{Edges $V_H \times \{s, t \}$, (\Cref{alg:q:line:new-h-edges})}}
         +
         \underbrace{n^{2 - 2\s} }
         _
         {\substack{
            \text{Dijkstra} \text{(\Cref{alg:q:line:dijkstra})} \\
            \text{and $n^{1-\s}\cdot n\log A$} \\
            \text{(\Cref{thm:ssp-undir})}}}
         + 
         \underbrace{
         n \cdot Q(1, 1) \cdot %
         \log W}
         _
         {\text{all ShortestPaths (\Cref{alg:q:line:g-copies})}}
     \right) 
 \end{align*}
  Where the complexity of \Cref{alg:q:line:g-copies} is bounded by 
 \Cref{thm:ssp-undir} and \Cref{not-too-many-total:undir}.
 
 Total update time (the lines are from \Cref{wda-final-update}): 
 \begin{align*}
     \tilde{O}
     \left( 
      ( 
          \underbrace
          {
              U_{1+\epsilon} + U
          }
          _
          {
              \text
              {(
                  \Cref{alg:upd:line:o1e},
                  \Cref{alg:upd:line:o}
              )}
            }   
      +
      \underbrace{
        Q_{1+\epsilon}(n^{1-\s}, n^{1-\s})
        }
      _
      {\text{edges of }H, \text{(\Cref{alg:upd:line:h})}}  
      )
     \log W
     \right) 
 \end{align*}
 
\end{proof}

\begin{theorem}[directed, exact]
    \label{thm:exact-dir-gen} 
    \thmstatement{$\mathcal{O}_{X}$ and $\mathcal{O}^{(2)}_{X}$}{directed}{$Wn^\s$}{integer}
     
    \begin{itemize}[nosep]
        \item \thmprep $\tilde O\left(P_{2} + P\right)$
        \item \thmupd $$
        \tilde O\left(
            U_{2} + U
            +
            Q(n^{1-\s}, n^{1-\s}) 
            + 
            Q_{2}(n^{1-\s}, n) 
        \right)
        $$
        
        \item  \thmquery{an exact $st$-shortest}
        $$
        \tilde O\left(  
                Q(1, n^{1-\s}) 
                +
                Q_{2}(1, n) 
                +
                n^{2 - \s} 
                +
                n \cdot Q(1, 1) \log W
        \right)
        $$
    \end{itemize}
    \thmps
\end{theorem}

\begin{proof}
Note that \Cref{construct-h} extends to $\epsilon=0$ in which case there exists only one $(A,B)$-rounded graph $G_0=G$.

Consider the algorithm defined by \Cref{wda-final-update} for update operation and \Cref{wda-final-query} for query operation (all lines are executed).

\paragraph{Correctness.} The fact that concatenated shortest subpaths form w.h.p.~an exact $st$-shortest path in $G$ follows from \Cref{hi-hj-aprx} and the correctness of each subpath follows from \Cref{thm:ssp-dir}.

\paragraph{Complexity.} Total query time (the lines are from \Cref{wda-final-query}): 
 \begin{align*}
     \tilde{O}\left(  
        \underbrace{
        Q_{}(1, n^{1-\s}) \cdot %
          }
         _
         {\text{Edges $V_H \times \{s, t \}$, (\Cref{alg:q:line:new-h-edges})}}
         +
         \underbrace{n^{2 - 2\s} }
         _
         {\substack{\text{Dijkstra} \\ \text{(\Cref{alg:q:line:dijkstra})}}}
         + 
         \underbrace{
         Q_2(n, 1)  %
          }
         _
         {\text{\Cref{alg:q:line:o2query}} }
         + 
         \underbrace{n^{2 - \s}  }
         _
         {\text{$P_i$'s, \Cref{alg:q:line:plausible}} }
         + 
         \underbrace{
         n \cdot Q(1, 1)  %
         }
         _
         {\text{all ShortestPaths (\Cref{alg:q:line:g-copies})}}
     \right) 
 \end{align*}
Note that the result of $queryAll(V,V_H)$ for \Cref{alg:q:line:plausible} was precomputed during the \textsc{Update} routine (\Cref{wda-final-update}), so the cost of that query does not occur here for \textsc{Query}.
  The complexity of \Cref{alg:q:line:g-copies} is bounded by 
 \Cref{thm:ssp-dir} and \Cref{not-too-many-total}.
 
 Total update time (the lines are from \Cref{wda-final-update}): 
 \begin{align*}
     \tilde{O}
     \left( 
          \underbrace{
           U_{2}  + U
          }
          _
          {
          \text{(\Cref{alg:upd:line:o2}, 
          \Cref{alg:upd:line:o})
          }}   
      +
      \underbrace{
        Q(n^{1-\s}, n^{1-\s})
        }
      _
      {\text{edges of }H, \text{(\Cref{alg:upd:line:h})}}
      +
      \underbrace{
      Q_2(n, n^{1-\s})
      }
      _
      {\mathcal{O}^{(2)}_{G_x}.queryAll(V, V_H), \text{(\Cref{alg:upd:line:o2qall}})}  
     \right) 
 \end{align*}
 
\end{proof}

\section{Applying the Blackbox Reductions}
\label{sec:fullalgo}

In this section, we prove our main results for weighted graphs, that is \Cref{thm:intro:directedintweighted,thm:intro:directedweighted} stated in the introduction.
The proofs for these results stem from the techniques in \Cref{sec:weighted}, which presents blackbox reduction from dynamic shortest path to dynamic distances.
We apply these reduction to the dynamic distance data structures constructed by \cite{Sankowski05,BrandN19}.
In \Cref{sec:oracles}, we state the update and query complexities of \cite{Sankowski05,BrandN19}. 
Then in the subsequent \Cref{sec:final}, we plug the complexities into our blackbox reductions from \Cref{sec:weighted}, resulting in \Cref{thm:intro:directedintweighted,thm:intro:directedweighted}.

\subsection{Oracles}
\label{sec:oracles}

Let us derive the time complexities for the dynamic distance oracles used in our data structure. We will need the following lemma.

\begin{lemma}[{\cite{Sankowski05}, \cite[Theorem 4.2]{BrandN19}}] \label{san-bn}
    For any $\epsilon>0$, $\s > 0$ and $0\le\mu\le1$, $0\le\nu\le1$ there exists a dynamic algorithm that maintains $(1+\epsilon)$-approximate $n^\s$-bounded distances in a positive integer weighted directed graph.
    The update time is
    $$
    \tilde{O}(
    n^{\omega(1,1,\nu+\s)-\nu}/\epsilon
    +n^{1+\mu+\s}+n^{\omega(1,1,\mu)-\mu+\s}
    ).
    $$
    The query time to query any pairwise $S\times T$ distances for any $|S|=n^{\delta_1},|T|=n^{\delta_2}$ (the sets are not fixed, but given when performing the query), is
    $$
    \tilde{O}(
    n^{\omega(\delta_1,\nu+\s,\delta_2)} / \epsilon
    ).
    $$
    The exact $n^\s$-bounded distance can be queried in
    $$
    \tilde{O}(
    n^{\omega(\delta_1,\mu,\delta_2)+\s}
    ).
    $$
\end{lemma}
Using this lemma and the fact that $\omega(a, b, c + d) \leq \omega(a, b, c) + d$ we can derive the needed time-complexities for update and query operations (\Cref{tab:aprx-complexities} and \Cref{tab:exact-complexities}).
Further, to query the $uv$-distance for any pair $u,v\in V$, we can pick $S=\{u\},T=\{v\}$ and get $\tilde O(n^{\omega(0,\nu+\s,0)}/\epsilon) = \tilde O(n^{\nu+\s}/\epsilon)$ query time for the approximate $uv$-distance and $\tilde{O}(n^{\omega(0,\mu,0)+\s}) = \tilde O(n^{\mu+\s})$ query time for the exact distance.

\begin{table}[t!]
    \centering
    \renewcommand{\arraystretch}{1.5}
    \begin{tabular}{|l||c|c|c|}
        \hline
       & $\mathcal{O}^{(2)}_{G_x}$  & $\mathcal{O}^{(1+\epsilon)}_{G_x}$ & $\mathcal{O}_{G_x}$  \\
         \hline
         \hline
        $update(*)$ & $ \tilde O( \OrAprxUpdNDep /  \epsilon )$ & $\tilde O( \OrAprxUpdNDep /  \epsilon^2 )$ & $\tilde O(n^{1+\mu+\s}/  \epsilon + n^{\omega(1, 1, \mu)-\mu + \s})$ \\
        \hline
        $query(u, v)$ & &  & $\tilde O(n^{\s + \mu} / \epsilon)$ \\
        \hline
        $queryAll(V_H, V_H)$ 
        & 
        & $\tilde O( \OrAprxQueryHHNDep / \epsilon^2  )$  & \\
        \hline
        $queryAll(V, V_H)$ & $\tilde O( \OrAprxQueryVHNDep /  \epsilon  )$  &  $\tilde O(\OrAprxQueryVHNDep  /  \epsilon^2  )$ & \\
        \hline
    \end{tabular}
    \caption{Complexities of $4n^\s / \epsilon$-bounded distance oracles on any graph $G_x$. Complexities that are not used in the final analysis are omitted. Here complexities of $\mathcal{O}^{(2)}_{G_x}$ scale in $\epsilon$, because we maintain $4n^\s / \epsilon$-bounded distances.}
    \label{tab:aprx-complexities}
~

    \centering
    \renewcommand{\arraystretch}{1.5}
    \begin{tabular}{|l||c|c|}
        \hline
       & $\mathcal{O}^{(2)}_{G_x}$  & $\mathcal{O}_{G_x}$  \\
         \hline
         \hline
        $update(*)$ & $\tilde O(n^{\omega(1, 1, \nu + \s)-\nu} {W} )$ &  $\tilde O(\OrExactUpdNDep W)$ \\
        \hline
        $query(*)$ &   & $\tilde O(n^{\s + \mu} {W})$ \\
        \hline
        $queryAll(V_H, V_H)$ 
        & 
        & 
        $\tilde O(n^{\omega(1-\s,1-\s,\mu ) + \s} {W})$  \\
        \hline
        $queryAll(V, V_H)$ & $\tilde O(n^{\omega(1,1-\s,\nu + \s) } {W})$  & \\
        \hline
    \end{tabular}
    \caption{Complexities of ${W}n^\s$-distance bounded oracles on graph $G$ with integer weights in $[1, W]$}
    \label{tab:exact-complexities}
\end{table}

\subsection{Final Complexities}
\label{sec:final}

Plugging, the complexities of \Cref{san-bn} into the blackbox reduction \Cref{thm:aprx-dir-gen} for approximate directed shortest paths, we obtain the following \Cref{thm:aprx-dir}.
\begin{corollary}[Approximate, directed]
    \label{thm:aprx-dir}
    For any $0\le \s\le 1$, $0\le\mu\le1$, $0\le\nu\le1$, $\epsilon>0$,
    there exists a fully dynamic algorithms that maintain $(1+\epsilon)$-approximate shortest paths
    for directed graphs with real edge weights in $[1, W]$.
    The preprocessing time is $O(n^{\omega+\s}\epsilon^{-1} \log W)$, 
    the update time for an edge insertion or deletion is
    $\tilde{O}(\thmAprxDirUpd)$ and
    querying the shortest path for any $s,t\in V$ takes
    $\tilde{O}(\thmAprxDirQuery
    )$ time.    
    The dynamic algorithm is randomized and correct w.h.p.~with one-sided error and works against an adaptive adversary.

\end{corollary}

Balancing the terms with parameters 
$$\s \approx \directedS, \nu \approx \directedNu, \mu \approx  \directedMu$$
the preprocessing time complexity is $\tilde O(n^{\directedPre} \log W)$, query time complexity is $\tilde O(n^{\directedQuery} \log W)$ and update time complexity is $\tilde O(n^{\directed} \log W)$ 
by current bounds on matrix multiplication \cite{Williams12,Gall14,AlmanW21,DuanWZ22,GallU18}\footnote{%
Parameters balanced via \cite{complexity}. This specific result is available \href{https://www.ocf.berkeley.edu/~vdbrand/complexity/index.php?terms=\%23\%20query\%2C\%20n\%5E0.075\%20faster\%20than\%20update\%0A2\%20-\%20s\%20\%2B\%200.075\%0A1\%20\%2B\%20nu\%20\%2B\%20s\%20\%2B\%200.075\%0A1\%20\%2B\%20mu\%20\%2B\%20s\%20\%2B\%200.075\%0A\%23\%20update\%0Aomega(1\%2C\%201\%2C\%20nu\%20\%2B\%20s)\%20-\%20nu\%0A1\%20\%2B\%20mu\%20\%2B\%20s\%0Aomega(1\%2C\%201\%2C\%20mu)\%20-\%20mu\%20\%2B\%20s\%0Aomega(1\%2C\%201\%20-\%20s\%2C\%20nu\%20\%2B\%20s)&a=1}{here}.}.
These are precisely the complexities stated in \Cref{thm:intro:directedweighted} for directed graphs.

\begin{proof}[Proof of \Cref{thm:aprx-dir}]
 
Using 
\begin{align*}
    Q_2(n^{1-\s}, n) &= \tilde O( \OrAprxQueryVHNDep /  \epsilon  ), 
    \\
    Q_{1+\epsilon}(n^{1-\s} ,n^{1-\s}) &= \tilde O( \OrAprxQueryHHNDep / \epsilon^2  ),
    \\
    Q( 1, 1) &= \tilde O(n^{\s + \mu} / \epsilon), 
    \\
    U_2 &= \tilde O( \OrAprxUpdNDep /  \epsilon ), 
    \\
    U_{1+\epsilon} &= \tilde O( \OrAprxUpdNDep /  \epsilon^2 ), 
    \\
    U &= \tilde O(n^{1+\mu+\s}/  \epsilon + n^{\omega(1, 1, \mu)-\mu + \s})
\end{align*}
we have the following complexities.

The preprocessing time is $O(n^{\omega+\s}\epsilon^{-1} \log W)$. 
By \Cref{thm:aprx-dir-gen}, the update time is
\begin{align*}
    &~\tilde O\left(\log W \cdot 
        \left(
            U_{2} + U_{1+\epsilon} + U
            +
            Q_{1+\epsilon}(n^{1-\s}, n^{1-\s}) 
            + 
            Q_{2}(n^{1-\s}, n) 
        \right)
    \right)\\
    =&~
    \tilde O\left(\thmAprxDirUpd
    \right)
\end{align*}
and the query time is
\begin{align*}
        &~
        \tilde O\left( \log W \cdot 
            \right(
                Q_{1+\epsilon}(1, n^{1-\s}) 
                + 
                Q_{2}(1, n) 
                +
                n^{2 - \s} 
                +
                n \cdot Q(1, 1)
            \left)
        \right)\\
        =&~
        \tilde O \left(
        \thmAprxDirQuery
        \right)
\end{align*}
\end{proof}

Using the blackbox reduction \Cref{thm:aprx-undir-gen} for approximate shortest paths on undirected graphs, together with the dynamic distance data structure from \Cref{san-bn}, we obtain \Cref{thm:aprx-undir}.

\begin{corollary}[Approximate, undirected]
\label{thm:aprx-undir}

    For any $0\le \s\le 1$, $0\le\mu\le1$, $0\le\nu\le1$, $\epsilon>0$,
    there exists a fully dynamic algorithm that maintain $(1+\epsilon)$-approximate shortest paths
    for undirected graphs with real edge weights in $[1, W]$.
    The preprocessing time is $O(n^{\omega+\s}\epsilon^{-1} \log W)$,
    the update time for an edge insertion or deletion is
    $\tilde{O}(\thmAprxUndirUpd)$,
    and querying the shortest path for any $s,t\in V$ takes
    $\tilde{O}(\thmAprxUndirQuery)$ time.    
    The dynamic algorithm is randomized and correct w.h.p., and works against an adaptive adversary.
\end{corollary}

Balancing the terms with parameters 
$$\s \approx 0.1909, \nu \approx 0.3376, \mu \approx  0.5286$$ the preprocessing time complexity is $\tilde O(n^{\undirectedPre})$
the query and update time complexities are both $\tilde O(n^{\undirected} \log W)$
by current bounds on matrix multiplication \cite{Williams12,Gall14,AlmanW21,DuanWZ22,GallU18}\footnote{
Parameters balanced via \cite{complexity}. This specific result is available 
\href{https://www.ocf.berkeley.edu/~vdbrand/complexity/index.php?terms=2\%20-\%202s\%0A1\%20\%2B\%20nu\%0A1\%20\%2B\%20mu\%20\%2B\%20s\%0A\%0Aomega(1\%2C\%201\%2C\%20nu\%20\%2B\%20s)\%20-\%20nu\%0A1\%20\%2B\%20mu\%20\%2B\%20s\%0Aomega(1\%2C\%201\%2C\%20mu)\%20-\%20mu\%20\%2B\%20s\%0Aomega(1-s\%2C\%201\%20-\%20s\%2C\%20nu\%20\%2B\%20s)&a=1}{here}.
}. This is the result stated in \Cref{thm:intro:undirectedweighted} for undirected graphs.

\begin{proof}[Proof of \Cref{thm:aprx-undir}]
Using 
\begin{align*}
    Q_{1+\epsilon}(n^{1-\s} ,n^{1-\s}) &= \tilde O( \OrAprxQueryHHNDep / \epsilon^2  ),
    \\
    Q( 1, 1) &= \tilde O(n^{\s + \mu} / \epsilon), 
    \\
    U_{1+\epsilon} &= \tilde O( \OrAprxUpdNDep /  \epsilon^2 ), 
    \\
    U &= \tilde O(n^{1+\mu+\s}/  \epsilon + n^{\omega(1, 1, \mu)-\mu + \s})
\end{align*}
we have the following complexities.

The preprocessing time is $O(n^{\omega+\s}\epsilon^{-1} \log W)$. 
By \Cref{thm:aprx-undir-gen}, the update time is
\begin{align*}
    &~\tilde O\left(\log W \cdot 
        \left(
            U_{1+\epsilon} + U
            +
            Q_{1+\epsilon}(n^{1-\s}, n^{1-\s}) 
        \right)
    \right)\\
    =&~
    \tilde O\left(\thmAprxUndirUpd
    \right)
\end{align*}
and the query time is
\begin{align*}
        &~
        \tilde O\left( \log W \cdot 
            \right(
                 Q_{1+\epsilon}(1, n^{1-\s}) + n^{2 - 2\s} + n \cdot Q(1, 1)
            \left)
        \right)\\
        =&~
        \tilde O \left(
        \thmAprxUndirQuery
        \right)
\end{align*}
\end{proof}

At last, we obtain a result for exact shortest paths on directed graphs.
For this we use the reduction of \Cref{thm:exact-dir-gen} together with the dynamic distance results from \Cref{san-bn}.

\begin{corollary}[Exact, directed]  \label{thm:exact-dir}
    For any $0\le \s\le 1$, $0\le\mu\le1$, $0\le\nu\le1$,
    there exists a fully dynamic algorithm that maintains the exact shortest paths
    for directed graphs with integer edge weights in $[1, W]$.
    The preprocessing time is $O(n^{\omega+\s}{W}\log W)$, 
    the update time for an edge insertion or deletion is
    $\tilde{O}(\thmExactDirUpd )$, and
    querying the shortest path for any $s,t\in V$ takes
    $\tilde{O}(\thmExactDirQuery 
    )$ time.
    The dynamic algorithm is randomized and correct w.h.p., and works against an adaptive adversary.
\end{corollary}

Balancing the terms with parameters 
$$\s \approx \exactS, \nu \approx \exactNu, \mu \approx \exactMu$$
the preprocessing time complexity is $\tilde O(n^{\exactPre})$, query time complexity is $\tilde O(n^{\exactQuery})$ and update time complexity is $\tilde O(n^{\exact})$ by current bounds on matrix multiplication \cite{Williams12,Gall14,AlmanW21,DuanWZ22,GallU18}\footnote{%
Parameters balanced via \cite{complexity}. This specific result is available 
\href{https://www.ocf.berkeley.edu/~vdbrand/complexity/index.php?terms=2\%20-\%20s\%2B0.076\%0A1\%20\%2B\%20nu\%20\%2B\%20s\%2B0.076\%0A1\%20\%2B\%20mu\%20\%2B\%20s\%2B0.076\%0A\%0Aomega(1\%2C\%201\%2C\%20nu\%20\%2B\%20s)\%20-\%20nu\%0A1\%20\%2B\%20mu\%20\%2B\%20s\%0Aomega(1\%2C\%201\%2C\%20mu)\%20-\%20mu\%20\%2B\%20s\%0Aomega(1\%2C\%201\%20-\%20s\%2C\%20nu\%20\%2B\%20s)\%0Aomega(1-s\%2C1-s\%2Cmu)\%2Bs&a=1}{here}.}.
This is precisely \Cref{thm:intro:directedintweighted} from the intro.

 \begin{proof}[Proof of \Cref{thm:exact-dir}]

Using 
\begin{align*}
    Q_2(n^{1-\s}, n) &= \tilde O( \OrAprxQueryVHNDep  W ), 
    \\
    Q( 1, 1) &= \tilde O(n^{\s + \mu} W), 
    \\
    U_2 &= \tilde O( \OrAprxUpdNDep W ), 
    \\
    U &= \tilde O(n^{1+\mu+\s} W + n^{\omega(1, 1, \mu)-\mu + \s})
\end{align*}
we have the following complexities.

The preprocessing time is $O(n^{\omega+\s}\epsilon^{-1} \log W)$. 
By \Cref{thm:exact-dir-gen}, the update time is
\begin{align*}
    &~\tilde O\left(\log W \cdot 
        \left(
            U_{2} + U
            +
            Q(n^{1-\s}, n^{1-\s}) 
            + 
            Q_{2}(n^{1-\s}, n) 
        \right)
    \right)\\
    =&~
    \tilde O\left(\thmExactDirUpd
    \right)
\end{align*}
and the query time is
\begin{align*}
        &~
        \tilde O\left( \log W \cdot 
            \right(
                Q(1, n^{1-\s}) 
                +
                Q_{2}(1, n) 
                +
                n^{2 - \s} 
                +
                n \cdot Q(1, 1) \log W
            \left)
        \right)\\
        =&~
        \tilde O \left(
        \thmExactDirQuery
        \right)
\end{align*}
\end{proof}

\section{Path Reporting on Unweighted Graphs}
\label{sec:unweighted}

In this section, we prove our results on unweighted undirected graphs, i.e.~\Cref{thm:intro:unweighted:st,thm:intro:unweighted:sssp}.
We show that we can deterministically maintain approximate shortest paths in subquadratic time.
The dynamic algorithms internally use fast rectangular matrix multiplication and their complexity can be parameterized by $\rho$, where $\rho$ is the solution to $\omega(1,1,\rho) = 1+2\rho$. Currently, $\rho \approx 0.529$ using the upper bounds on rectangular matrix multiplication by Le Gall and Urrutia \cite{GallU18}.
The exact statements proven in this section are given by \Cref{thm:unweighted}.

\begin{theorem}[Undirected, Unweighted, Approximate]\label{thm:unweighted}
There exist the following fully dynamic algorithms that maintain $(1+\epsilon)$-approximate shortest paths
for unweighted undirected graphs.
\begin{enumerate}[nosep]
    \item \label{item:unweighted:st}
    A dynamic algorithm with 
    $O(n^{1+\rho}\epsilon^{-2}\log\epsilon^{-1})$ update and query time 
    to return a $(1+\epsilon)$-approximate $st$-shortest path for any $s,t\in V$. 
    The preprocessing time is $O(n^{\omega}\epsilon^{-2}\log\epsilon^{-1})$.
    \item \label{item:unweighted:singlesource}
    A dynamic algorithm with $O(n^{(3+\rho)/2}\epsilon^{-2}\log\epsilon^{-1})$ update and query time 
    to return a $(1+\epsilon)$-approximate single source shortest paths tree for any $s\in V$.
    The preprocessing time is $O(n^{\omega+(1-\rho)/2}\epsilon^{-2}\log\epsilon^{-1})$.
    \item \label{item:unweighted:additive}
    If we also allow for $+4$ additive error, then there is a dynamic algorithm with 
    $O(n^{1+\rho}\epsilon^{-2}\log\epsilon^{-1})$ update and query time to return an approximate single source shortest path tree $T$ for any $s\in V$. Here $\dist_T(s,v) \le (1+\epsilon)\dist_G(s,t)+4$. 
    The preprocessing time is $O(n^{\omega}\epsilon^{-2}\log\epsilon^{-1})$. 
\end{enumerate}
All these dynamic algorithms are deterministic.
\end{theorem}

Our algorithms build on the deterministic dynamic algorithm by v.d.Brand, Forster and Nazari, \cite{BrandFN22} which could maintain approximate distances but not the respective paths.
Their dynamic algorithm maintains a $(1+\epsilon, 4)$-emulator of the input graph $G$. Such an emulator is a graph $H$ on the same vertex set as $G$ with the property $\dist_G(s,t) \le \dist_H(s,t) \le (1+\epsilon)\dist_G(s,t) + 4$ for all $s,t \in V$. Note that $H$ is not a subgraph of $G$, i.e.~it can contain edges that do not exist in $G$.
So while running Dijkstra's algorithm on $H$ returns good approximations of the distances in $G$, it does not return approximately the shortest paths in $G$.
Our dynamic algorithms from \Cref{thm:unweighted} work by replacing edges from $H$ by short paths in $G$. This way, we can transform the shortest path in $H$ into an approximately shortest path in $G$.

We use the following \Cref{lem:unweighted:emulator_maintenance} from \cite{BrandFN22} to maintain the emulator $H$.

\begin{algorithm2e}
\caption{Emulator construction from \cite{BrandFN22}\label{alg:unweighted:emulator}.} 
\SetKwProg{Proc}{procedure}{}{}
\Proc{$\textsc{Emulator}(G=(V,E), S \subset V)$}{
$S$ is a subset such that each vertex $v\in V$ with $\deg_G(v) > \sqrt{n\log n}$ has a neighbor in $S$.\\
$H = (V, E')$ will be the emulator. Initialize $E' = \emptyset$. \\
Add each $\{u,v\} \in E$ to $E'$ if $\min\{\deg_G(u), \deg_G(v)\} \le \sqrt{n}$ \\
For each pair $u,v\in S$ with $\dist_G(u,v) \le \lceil 4/\epsilon\rceil+2$, add edge $\{u,v\}$ to $E'$ with edge cost $\dist_G(u,v)$ \label{line:shortdist}\\
\Return $H = (V,E')$
}
\end{algorithm2e}

\begin{lemma}[{\cite[Section 3.2 for $d=\sqrt{n}$]{BrandFN22}}]\label{lem:unweighted:emulator_maintenance}
Given an unweighted graph $G = (V, E)$, $0 < \epsilon < 1$, 
we can deterministically maintain a $(1+\epsilon, 4)$-emulator with size $O(n^{3/2}\sqrt{\log n})$. 
The worst-case update time is $O((n^{\omega(1,1,\mu)-\mu} +n^{1+\mu} )\epsilon^{-2} \log \epsilon^{-1})$ for any $0 \le \mu \le 1$ and preprocessing time is $O(n^{\omega}
\epsilon^{-2} \log \epsilon^{-1})$.

The dynamic algorithm internally runs \Cref{alg:unweighted:emulator} to construct the emulator
and maintains pairwise $(\lceil4/\epsilon\rceil+2)$-bounded distances of $S\times V$ for the set $S\subset V$ used in \Cref{alg:unweighted:emulator}.
\end{lemma}

We now prove in \Cref{lem:unweighted:additive} that we can use \Cref{lem:unweighted:emulator_maintenance} to obtain approximately shortest paths in $G$, by replacing some of the edges in $H$ with short paths in $G$. Note that \Cref{lem:unweighted:additive} proves \cref{item:unweighted:additive} of \Cref{thm:unweighted}.

\begin{lemma}[{\Cref{item:unweighted:additive} of \Cref{thm:unweighted}}]\label{lem:unweighted:additive}
    Given an unweighted graph $G = (V, E)$, $0 < \epsilon < 1$, 
    we can deterministically maintain approximate single source shortest paths.
    The worst-case update time is $O((n^{\omega(1,1,\mu)-\mu} +
    n^{1+\mu} )\epsilon^{-2} \log \epsilon^{-1})$ for any $0 \le \mu \le 1$ and preprocessing time is $O(n^{\omega}
    \epsilon^{-2} \log \epsilon^{-1})$.

    A query receives any $s\in V$ and after $O(n^{1.5}\log^{1.5} n + n/\epsilon)$ time returns an approximate shortest path tree $T$ with $\dist_T(s,v) \le (1+\epsilon)\dist_G(s,v) + 4$. 
\end{lemma}

\begin{proof}
    We run the dynamic emulator algorithm from \Cref{lem:unweighted:emulator_maintenance} and let $H$ be the maintained emulator.
    During a query, we run Dijkstra's algorithm from $s$ on emulator $H$.
    This gives us a shortest paths tree $T$ on $H$ rooted at $s$.
    Some of the edges in $T$ might not exist in the original graph $G$ as $H$ is an emulator.
    We will replace each of these edges $\{u,v\}$ in $T$ by the $uv$-shortest path from $G$, via a routine we describe later. Let $T'$ be the resulting graph.
    
    Note that by line \ref{line:shortdist} in \Cref{alg:unweighted:emulator}, we have that $\dist_G(u,v) \le \lceil 4/\epsilon\rceil+2$.
    Thus after replacing edges $\{u,v\}$ with these shortest paths, the graph $T'$ has at most $O(n/\epsilon)$ many edges.
    So we can run BFS on this modified tree $T'$ in $O(n/\epsilon)$ time to get an approximate shortest path tree on $G$.

    We are left with explaining how to replace an edge $\{u,v\}$ in $T$ by $uv$-shortest paths.
    
    \paragraph{Replacing the edges}
    
    Consider an edge $\{u,v\}$ that we want to replace in $T$ by some $uv$-shortest path in $G$.
    Note that we only replace edges $\{u,v\}$ that exist in $H$ but do not exist in $G$.
    These edges have $u,v\in S$ (by \Cref{line:shortdist}) where $S\subset V$ is the set from \Cref{alg:unweighted:emulator}.
    Further, the $uv$-shortest path has length at most $\lceil4/\epsilon\rceil+2$ (by \Cref{line:shortdist})
    and \Cref{lem:unweighted:emulator_maintenance} maintains all distances of pairs $S\times V$.
    So we can sort all vertices $V$ based on their distance to $u$.
    Then we iterate as follows:
    for every vertex $w \in V$ with $\dist(w,v)=\dist(u,v)-1$, check if edge $\{w,v\} \in E$. If yes, add that edge to $T$ and recurse on finding the shortest $w,v$ path.

    In total this takes $O(n \log n)$ time to replace one $\{u,v\}$ edge by a $uv$-shortest path, because we sort the vertices only once and then iterate over each vertex at most once to check if it's on the $uv$-shortest path.
    
    \paragraph{Bounding the number of replaced edges}
    The previous paragraph showed that any one edge $\{u,v\}$ in $T$ can be replaced by a $uv$-shortest path in $O(n\log n)$ time.
    To bound the total time, we are left with bounding how many edges in $T$ must be replaced.
    We show there are at most $\tilde O(\sqrt{n})$ edges in $T$ that we must replace by short paths.
    
    Note that we only replace edges $\{u,v\}$ that exist in $H$ but do not exist in $G$.
    These edges have $u,v\in S$ (by \Cref{line:shortdist}) where $S\subset V$ is the set from \Cref{alg:unweighted:emulator}.

    Now for sake of analysis, assume the tree $T$ is directed with the edges oriented away from the source vertex $s$. We can assume this, since $T$ is a shortest path tree rooted at $s$. 
    Since it's a tree and not a DAG, each $v \in S$ has at most one incoming edge in $T$.
    Thus there are at most $|S| = \sqrt{n\log n}$ edges in $T'$ that do not exist in $G$.

    Thus, replacing all edges in $T$ that do not exist in $G$ takes $O(|S|\cdot n\log n) = O(n^{1.5} \log^{1.5} n)$ time.

    \paragraph{Summary}
    The update time is $O(n^{1+\rho}\epsilon^{-2} \log \epsilon^{-1})$ as we run \Cref{lem:unweighted:emulator_maintenance} to maintain the emulator.
    
    The query time is $O(n^{1.5}\log^{1.5} n + n/\epsilon)$ as the first term bounds the size of the emulator on which we run BFS, and the time to replace the necessary edges in $T$ to obtain $T'$.
    The second term bounds the time to run BFS on $T'$.
\end{proof}

The additive $+4$ error of the shortest paths tree maintained in \Cref{lem:unweighted:additive} only matters for pairs $u,v$ with $\dist(u,v) < 4\epsilon$ as otherwise the additive error is just a multiplicative $(1+\epsilon)$ error.
To obtain the distances for these pairs where $\dist(u,v) < 4\epsilon$, we can use the following \Cref{lem:unweighted:boundedsinglesource}.

\begin{lemma}[\cite{BrandFN22,Sankowski05}]\label{lem:unweighted:boundedsinglesource}
    There exists a deterministic fully dynamic algorithm that maintains exact $h$-bounded single-source distances
    for unweighted undirected graphs.
    The update time and query time for any $s$ is $\tilde{O}((n^{1+\rho})h^2 \log h)$ and the preprocessing time is $O(n^{\omega}h^2 \log h)$.
\end{lemma}

Using the distances maintained via \Cref{lem:unweighted:boundedsinglesource},
we now want to reconstruct the shortest paths for pairs $s,t$ with $\dist(s,t) < 4\epsilon$.
The techniques used to reconstruct the paths are by Karczmarz, Mukherjee and Sankowski \cite{KarczmarzMS22} which they used to maintain reachability with path reporting on DAGs.

\begin{lemma}[{\Cref{item:unweighted:st} of \Cref{thm:unweighted}}]\label{lem:unweighted:st}
    Given an unweighted graph $G = (V, E)$, $0 < \epsilon < 1$, and $0\le\mu\le 1$,
    we can deterministically maintain approximate shortest paths.
    
    The worst-case update and query time is $O((n^{1+\rho} )\epsilon^{-2} \log \epsilon^{-1})$.
    The query returns for any given $s,t$ a $(1+\epsilon)$-approximate $st$-shortest path.
    The preprocessing time is $O(n^{\omega}
    \epsilon^{-2} \log \epsilon^{-1})$.
\end{lemma}

\begin{proof}%
    We describe how to extend \cref{lem:unweighted:additive} to handle short paths.
    The issue of \cref{item:unweighted:additive} is that it has an additive $+4$ error wrt.~the distance.
    This only matters for short paths of length $4/\epsilon$. 
    
    We run \Cref{lem:unweighted:boundedsinglesource} to check if the $st$-distance is less than $4/\epsilon$.
    This takes $O(n^{1+\rho}\epsilon^{-2}\log \epsilon^{-1})$ per update and query.
    We also run \cref{item:unweighted:additive} which, if the $st$-distance is larger than $4/\epsilon$, yields a $(1+\epsilon)$-approximate $st$-shortest path.
    If the $st$-distance is less than $4/\epsilon$, we construct an exact $st$-shortest path via the path reporting approach by Karczmarz, Mukherjee and Sankowski \cite{KarczmarzMS22} as follows.

    \paragraph{Finding the $st$-shortest path}
    Via \Cref{lem:unweighted:boundedsinglesource}, we get the exact $4/\epsilon$-bounded single source distances for $s$ in $O(n^{1+\rho}\epsilon^{-2}\log \epsilon^{-1})$ time.
    Then we sort the vertices $v\in V$ descending by their distance $\dist_G(s,v)$.
    We set $v_{last}=t$ and iterate over the list and check for $v\in V$ if $\dist_G(s,v) = \dist_G(s,v_{last})-1$.
    In that case $(v_{last},t)$ must be the tail of the $st$-shortest path.
    We set $v_{last} \gets v$ and continue iterating over $V$.
    This takes $O(n)$ time in total.
\end{proof}

\begin{lemma}[{\Cref{item:unweighted:singlesource} of \Cref{thm:unweighted}}]\label{lem:unweighted:singlesource}
    Given an unweighted graph $G = (V, E)$, $0 < \epsilon < 1$, and $0\le\mu\le 1$,
    we can deterministically maintain approximate shortest paths.
    
    The worst-case update and query time is $O((n^{(3+\rho)/2} )\epsilon^{-2} \log \epsilon^{-1})$.
    The query returns for any given $s$ a $(1+\epsilon)$-approximate single source shortest path tree rooted at $s$.
    The preprocessing time is $O(n^{\omega}
    \epsilon^{-2} \log \epsilon^{-1})$.
\end{lemma}

\begin{proof}%
    We now describe how to extend \cref{item:unweighted:additive} to  \cref{lem:unweighted:singlesource}.

    The issue of \cref{item:unweighted:additive} is that it has an additive $+4$ error wrt.~the shortest path.
    This only matters for short paths of length $4\epsilon$. 
    Thus we can use the following approach:
    (i) run \cref{item:unweighted:additive} to obtain some tree $T$, (ii) also construct a single source shortest paths tree $T'$ truncated to depth $4\epsilon$.
    Then we take the union $T$ and $T'$.
    Note that we have $\dist_{T\cup T'}(s,v) \le \dist_G(s,v) \le (1+\epsilon) \dist_{T\cup T'}(s,v)$, because $T$ contains all edges of short paths, and $T'$ contains all edges to approximate long paths.

    To reduce $T\cup T'$ to a single tree, we run BFS from $s$ and return the shortest path tree constructed by the BFS. This final tree is then a $(1+\epsilon)$-approximate shortest path tree rooted at $s$.
    
    Finding this tree $T'$ uses techniques from the path reporting data structure by Karczmarz, Mukherjee and Sankowski \cite{KarczmarzMS22}.

    \paragraph{Finding the shortest paths tree $T'$}
    Construct graphs $G_1,...,G_p$ as follows:
    Assume the vertices of $G$ are $V = \{v_1,...,v_n\}$. 
    Graph $G_\ell$ contains $G$ and additional $2$ copies $v_i', v_i''$ of each vertex $v_i \in V$.
    Further, for each $\{v_i,v_j\}$ and $\{v_j,v_k\}$ in $G$ 
    with $j \in \{\ell n /p,..., (\ell+1)n / p\}$, 
    the graph $G_\ell$ also has edges $\{v_i,v_j'\}$ and $\{v_j', v_k''\}$.
    
    Observe that for any $v_i,v_j\in V$ an $v_iv_j''$-path exists in $G_\ell$ if an only if the last vertex visited before $v_j$ by the path has an index in $\{\{\ell n /p, (\ell+1)n / p\}$.

    We can now reconstruct the shortest paths tree rooted at $s$ as follows:
    Compute the single source distances rooted at $s$ in $G$ and each $G_1,...,G_p$.
    For each vertex $t\in V$ we do the following.
    Go through the $G_\ell$ for $\ell=1,...,p$ to check if $\dist_G(s,t)=\dist_{G_\ell}(s,t'')$. 
    If the euqlity holds, then we know there is an $st$-shortest path in $G$ with the last vertex visited before $t$ being a $v_j$ for $j \in \{\ell n /p,..., (\ell+1)n / p\}$.
    So iterate over $j \in \{\ell n /p,..., (\ell+1)n / p\}$ and check if $\dist_{G}(s,v_j) = \dist_G(s,t)-1$ and $\{v_j,t\} \in E$. When we find such a vertex, add edge $\{v_j,t\}$ to $T'$.
    At the end of this procedure, $T$ is a shortest paths tree rooted at $s$.

    It takes $O(n^{1+\rho}\epsilon^{-2}\log(\epsilon^{-1})\cdot p)$ time to maintain the single source distances for each $G_\ell$.
    Then constructing the paths with above procedure takes an extra $O(n\cdot (n/p + p))$ time, 
    as for each vertex $t\in V$ we (i) check the distance $\dist_{G_\ell}(s,t)$ for each $G_\ell$, and (ii) check the distance $\dist_{G_\ell}(s,v_j)$ for one $G_\ell$ and $n/p$ many $j$.

    By picking $p = n^{(1-\rho)/2}$ we get a complexity of
    \begin{align*}
    O(n^{1+\rho}\epsilon^{-2}\log(\epsilon^{-1}) n^{(1-\rho)/2} + n(n^{1-(1-\rho)/2} + n^{(1-\rho)}) )
    =&~
    O(
    n^{(3+\rho)/2}\epsilon^{-2}\log(\epsilon^{-1})
    +
    n^{2-\rho}) \\
    =&~
    O(n^{(3+\rho)/2} \epsilon^{-2}\log(\epsilon^{-1}))
    \end{align*}
    where the last equality used that $0.5 \le \rho$.
\end{proof}

\newpage
\printbibliography[heading=bibintoc]
\newpage
\appendix

\section{Auxiliary Graph Proofs}
\label{appendix}
The following is the hitting set argument from \cite{UllmanY91}, commonly used in graph algorithms.

\begin{lemma}\label{segmentation}

Let $G=(V, E)$ be an $n$-node graph and $h \in N$. 
Let $R \subset V$ be a random sample of size $\Omega((n / h)\log n)$. 
Then w.h.p~%
we have the following: 
For every $s, t \in V$ with $\dist(s, t) \geq h$ there is a shortest $st$-path that can be split into segments 
$s \path r_1 \path r_2 \path \ldots \path  r_d \path t$ 
where each $r_i \in R$ and each segment has at most $h$ hop. 

\end{lemma}

Before proving \Cref{hi-hj-aprx}, we argue that the edge weights of $H$ as constructed in \Cref{construct-h} corresponds to short paths in $G$.
\begin{corollary}\label{len-H-dist-G}
    For any pair $u, v \in V_H$ we have $\dist_G(u, v) \leq \len_H(u, v)$ and if the shortest $uv$-path in $G$ uses at most $n^\s$ hop, then we also have
    \begin{align*}
      \len_H(u, v) \leq  (1 + O(\epsilon)) \dist_G(u, v)
    \end{align*}
\end{corollary}

\begin{proof}
    \begin{align*}
        \len_H(u, v) = \min_x \frac{B_x}{A} \cdot \Delta_x(u,v)%
        \leq (1 + \epsilon) \cdot \min_x \frac{B_x}{A} \cdot \dist_{G_x}(u, v)
        \leq (1 + \epsilon)^2 \dist_G(u, v)
    \end{align*}
    where the first step uses the definition of $\len_H(u, v)$,
    the second step uses the definition of $\Delta_x$ and \Cref{segmentation}, %
    and the last step uses \Cref{AB-min-dist}.
\end{proof}

This directly implies $\dist_G(u,v)\le\dist_H(u,v)$ for all $u,v\in V_H$. 
We now prove the converse statement (up to approximation error).

\hihjaprx*
\begin{proof}[Proof of \Cref{hi-hj-aprx}]
    By the random construction of $V_H$ and \Cref{segmentation} 
    there exist a shortest $uv$-path in $G$ that has the form $u = h_0 \path h_1 \path ... \path h_d = v$
    where each $h_k \in V_H$ and the path $h_k \curly r_{k+1}$ has at most $n^\s$ edges in $G$.
    
    By \Cref{len-H-dist-G} we have $\len_H(h_k, h_{k+1}) \le (1 + \epsilon) \dist_G(h_k, h_{k+1})$ for all $k$, therefore
    \begin{align*}
        \dist_H(u, v) \le \sum_{k=0}^{d} \len_H(h_k, h_{k+1}) \le 
        \sum_{k=0}^{d} (1 + \epsilon) \dist_G(h_k, h_{k+1}) = (1 + \epsilon) \dist_G(u, v)
    \end{align*}
    On the other side, we have $\dist_G(u,v) \le \dist_H(u,v)$, because all edges in $H$ correspond to paths in $G$ and these edges have upper bounds on the respective path-length as edge weight.
\end{proof}

\end{document}